\documentclass[11pt,a4paper]{article}

\usepackage{hyperref}
\usepackage{a4wide}
\usepackage{amsmath,amsthm,amssymb}
\usepackage[latin1]{inputenc}

\usepackage{amsmath,amsfonts}
\usepackage{latexsym}

\newcommand{\localitemlabels}{
  \renewcommand{\theenumi}{(\roman{enumi})}
  \renewcommand{\labelenumi}{\theenumi}
}

\def \cbb{\mathbb{C}}
\def \nbb{\mathbb{N}}
\def \rbb{\mathbb{R}}

\def \zbb{\mathbb{Z}}

\def \ccal {\mathcal{C}}
\def \dcal {\mathcal{D}}
\def \ecal {\mathcal{E}}

\def \hcal {\mathcal{H}}

\def \kcal {\mathcal{K}}

\def \ocal {\mathcal{O}}

\def \qcal {\mathcal{Q}}

\def \scal {\mathcal{S}}

\def \wcal {\mathcal{W}}

\def \afk  {\mathfrak{A}}
\def \bfk  {\mathfrak{B}}

\def \<  {\langle}
\def \>  {\rangle}

\def \. { \,\! }

\def \cdotarg { \, \cdot \, }

\def \im  {\mathrm{Im}\,}

\def \idop {\mathbf{1}}

\def \st {^\ast}

\def \restrict {\lceil}

\def \supp {\mathrm{supp}\,}

\def \expltext#1 {\\ \text{\footnotesize{ (#1) }}\\}
\def \intercomm#1 {\\ \text{\footnotesize{ (#1) }}\\}
\def \undercomm#1 {\underset{\text{\scriptsize{ (#1) }}}}
\def \overcomm#1 {\overset{\text{\scriptsize{ (#1) }}}}

\def \boundedops {\bfk(\hcal)}

\def \tracenorm#1 { \| #1 \| _1 }
\def \hsnorm#1 { \| #1 \| _2 }

\def \vectorcomp#1 {
  \left( \begin{array}{c}
  #1
  \end{array}
  \right)  }

\def \pder#1#2 { \frac{ \partial #1 }{ \partial #2 }}

\newcommand{\sobnorm}[3]{ \| #1 \|_{#2,#3} }
\newcommand{\lsobnorm}[4]{ \| #1 \|^{(#2)}_{#3,#4} }

\newcommand{\convolmap}[1]{ \kappa_0 \lbrack #1 \rbrack }
\newcommand{\convolintmap}[1]{ \kappa \lbrack #1 \rbrack }
\newcommand{\convol}[2]{ \convolmap{#2} (#1)}
\newcommand{\convolint}[2]{ \convolintmap{#2} (#1)}

\newcommand{\syp}{\diamond}

\DeclareMathOperator{\img}{img}

\newcommand{ \cinftyh } {  \ccal^\infty (\hcal) }
\newcommand{ \cinftys } {  \ccal^\infty(\Sigma) }
\newcommand{ \cinftyss } {  \ccal^\infty(\Sigma)^\ast }
\newcommand{ \sigspace }[1]{ \ccal^{#1}(\Sigma) }

\newcommand{\lnorm}[2]{  \| #1 \|^{( #2 )}  }
\newcommand\PhiFH { \Phi_{ \mathrm{FH} }  }

\newcommand \PhiProdLoc { \Phi_\mathrm{prod,loc}   }
\newcommand \PhiProd { \Phi_{ \mathrm{prod} }  }
\newcommand \PhiProdL[1] { \Phi_{ \mathrm{prod} }^{#1}  }
\newcommand \PhiPos { \Phi_{ \mathrm{pos} }  }

\newcommand{\lcssample}{ \psi }
\newcommand{\psmap}{\Xi}

\newcommand{\wickprod}[1] {  : \!\! #1 \!\! : }

\newcommand{\oprod} {  \Pi }

\newcommand{\pd}[1]{ T_{#1} } 
\newcommand{\opd} { \pd{\Pi} } 

\def \cbb{\mathbb{C}}
\def \nbb{\mathbb{N}}
\def \rbb{\mathbb{R}}

\def \zbb{\mathbb{Z}}

\def \ccal {\mathcal{C}}
\def \dcal {\mathcal{D}}
\def \ecal {\mathcal{E}}

\def \hcal {\mathcal{H}}

\def \kcal {\mathcal{K}}

\def \ocal {\mathcal{O}}

\def \qcal {\mathcal{Q}}

\def \scal {\mathcal{S}}

\def \wcal {\mathcal{W}}

\def \afk  {\mathfrak{A}}
\def \bfk  {\mathfrak{B}}

\def \cdotarg { \, \cdot \, }

\def \im  {\mathrm{Im}\,}

\def \< {\langle}
\def \> {\rangle}

\def \st {^\ast}
\def \restrict {\lceil}

\def \boundedops {\bfk(\hcal)}

\def \tracenorm#1 { \| #1 \| _1 }
\def \hsnorm#1 { \| #1 \| _2 }

\newcommand{\bigsetprop}[2]{ \big\{ #1 \, \big| \, #2 \big\} }

\newcommand{\epe}{(1 \! + \! E)}

\newtheorem{definition}{Definition}
\newtheorem{lemma}[definition]{Lemma}
\newtheorem{proposition}[definition]{Proposition}
\newtheorem{theorem}[definition]{Theorem}
\newtheorem{corollary}[definition]{Corollary}

\numberwithin{equation}{section}
\numberwithin{definition}{section}

\newcommand{\inst}[1]{$^\textrm{#1}$ }
\newcommand{\email}[1]{e-mail: #1}

\newenvironment{proofwithremark}[1]{\begin{proof}[Proof (#1)]}{\end{proof}}
\newcommand{\cmpqed}{}

\begin{document}

\title{Quantum Inequalities from Operator Product Expansions}
\author{Henning Bostelmann\inst{1}\thanks{%
Supported by the EU network ``Noncommutative Geometry'' (MRTN-CT-2006-0031962)} 
\and Christopher J. Fewster\inst{2}}

\date{%
\parbox[t]{0.9\textwidth}{\footnotesize{%
\begin{enumerate}
\renewcommand{\theenumi}{\arabic{enumi}}
\renewcommand{\labelenumi}{\theenumi}
\item Dipartimento di Matematica, Università di Roma ``Tor Vergata'',
Via della Ricerca Scientifica, 00133 Roma, Italy,
\email{bostelma@mat.uniroma2.it}
\item Department of Mathematics, University of York,
Heslington, York YO10 5DD, United Kingdom, \email{cjf3@york.ac.uk}
\end{enumerate}
}}
\\
December 27, 2008}

\maketitle

\begin{center}
{\em Dedicated to the memory of Bernd Kuckert.}
\end{center}

\begin{abstract}
Quantum inequalities are lower bounds for local averages of quantum observables that have positive classical counterparts,
such as the energy density or the Wick square. 
We establish such inequalities in general (possibly interacting) quantum field theories on Minkowski space, 
using nonperturbative techniques. 
Our main tool is a rigorous version of the operator product expansion.
\end{abstract}

\section{Introduction}

The principal qualitative difference between classical and quantum physics
lies in the fundamentally unsharp nature of the latter, quantitatively
expressed by the uncertainty principle. This distinction becomes
particularly acute when one seeks analogues in quantum theory of quantities
that are classically positive. In quantum mechanics, for example, one
replaces a probability distribution over classical phase space by the Wigner function,
which is pointwise positive only for Gaussian states \cite{Hud:Wigner}. Consequently, Weyl quantization of classically positive observables
does not generally yield positive operators. Similarly, a positive (local) quadratic form in a
classical field and its derivatives, such as the energy density of a free
minimally coupled scalar field, would not be expected to have a positive
analogue in quantum field theory, owing to the subtractions necessary to
renormalize products of fields at a point. 

Nonetheless, positivity is not completely destroyed in quantization. The
sharp G{\aa}rding inequalities \cite{FefPho:gaarding_ineq} show that classically
positive symbols have Weyl quantizations that are positive modulo corrections of
lower order; that is, operators corresponding to a lower rate of growth in
momentum. The aim of this paper is to establish analogous results for
quantum field theory in a model independent and nonperturbative setting.
The key to our approach is a recently-developed \emph{microscopic phase space
condition} \cite{Bos:short_distance_structure} that controls the degrees of freedom available to the
theory at small scales and bounded energy, and guarantees the existence of a
rigorous operator product expansion (OPE) \cite{Bos:product_expansions}.
In any theory obeying this condition (along with other standard criteria
set out in Sec.~\ref{sec:outset}) we identify a class of `classically
positive' operator products
and show how this classical positivity is reflected in estimates on
suitable smearings of the composite fields appearing in the
corresponding OPEs. If
there is a distinguished normal product associated with the underlying
classically positive expression the picture is closely analogous to that
emerging from the G{\aa}rding inequalities: suitable smearings of the
normal product are positive modulo corrections of a lower order. As we
will describe, our results significantly generalize the quantum (energy) inequalities,
developed over recent years, that provide lower bounds on
smearings of quadratic normal ordered quantities in free field theories.  

In the following subsections, we will describe the background and
motivation for our study.

\subsection{Quantum inequalities}\label{sec:QI_background}

It has been known for many years that expectation values of quantities
such as the Wick square or energy density of a free scalar field may
assume negative values and are pointwise unbounded from below as the
quantum state is varied. Indeed, no local observable (other than
the zero operator) can be both positive and have a vanishing vacuum
expectation value \cite{EGJ:nonpos}. Thirty years ago, Ford made the key
observation that, as unrestricted negative energy densities or fluxes
could produce macroscopic violations of the second law of thermodynamics,
it was to be expected that QFT itself places strict limits on such
departures from positivity \cite{For:quantum_coherence}. Subsequently,
Ford and Roman were able to derive lower bounds, called quantum
inequalities (QIs), on averaged energy densities for scalar fields in
Minkowski space
\cite{For:neg_energy_flux,ForRom:averaged_energy,ForRom:restrictions_flat}; 
these results were generalized to static curved
spacetimes by Pfenning and Ford \cite{PfeFor:qei_static}. 

In the results just mentioned, the averaging is performed along a
timelike geodesic with respect to a Lorentzian weight. With Eveson, one
of us (CJF) obtained similar results for general weight functions
\cite{FewEve:qei_flat}. As an example, the renormalized energy density
$\wickprod{\rho}$ of the field of mass $m$ in
four-dimensional Minkowski space obeys the inequality
\begin{equation}\label{eq:example_qei}
\int dt\, \omega(\wickprod{\rho}(t,0)) |g(t)|^2 \ge -\qcal[g]:= -\frac{1}{16\pi^3}\int_m^\infty
du\, u^4 |\tilde{g}(u)|^2
\end{equation}
for any $g\in \dcal(\rbb)$ and all Hadamard states $\omega$ (this is slightly weaker than the bound
of \cite{FewEve:qei_flat}). Here $\tilde{g}$ denotes the Fourier
transform. Similar bounds are obeyed by any classically positive
field of form $\sum_i \wickprod{(P_i \phi)^2}$, where the $P_i$ are
partial differential operators with smooth real coefficients. We will
understand the term `quantum inequality' to apply to any bounds of this
type, and not just those relating to the energy density (for which the
more specific term `quantum energy inequality' (QEI) is also used).

The basic technique of \cite{FewEve:qei_flat} generalizes
straightforwardly to static spacetimes \cite{FewTeo:qei_static}
and the electromagnetic field \cite{Pfe:qei_em}. It also underlies the
general and rigorous results of \cite{Few:general_wordline}, which give
QIs for averaging with arbitrary weights along arbitrary timelike curves
in arbitrary globally hyperbolic spacetimes, valid for all Hadamard
states. (The bound in \cite{Few:general_wordline} is expressed using a
reference state; see \cite{FewSmi:qei_abs} for analogous results with a purely
local geometric bound.) Similar results hold for spin-$1$ fields
\cite{FewPfe:qei_spin1}. We note that averaging in timelike directions is
essential for establishing inequalities; while averaging over spacetime
volumes also yields lower bounds (see, e.g., \cite{FewSmi:qei_abs}), purely
spatial \cite{FRH:spatial_average} or lightlike \cite{FewRom:null} averaging is known not to be sufficient for
quantum inequalities.

An important feature of the lower bound in \eqref{eq:example_qei} is that it is independent
of the state $\omega$, and can be rewritten as an operator inequality
$\wickprod{\rho}(|g|^2)\ge -\qcal(g)\idop$. One cannot expect bounds of this
type for general interacting theories \cite{OluGra:static_domain_wall} (although
they {\em do} hold for conformal field theories in two dimensions
\cite{FewHol:qei_conformal}; see also
\cite{Fla:qi_twodimensional,Vol:qi_conformal} for precursors). Indeed,
the nonminimally coupled scalar field provides an example of a free
field theory in which averaged energy densities are unbounded from
below \cite{FewOst:qei_nonmin}. The best that can be expected, in
general, is an inequality of the form
$\wickprod{\rho}(|g|^2)\ge -\qcal(g)$, where $\qcal(g)$ is now permitted to be an
operator. As noted in \cite{Few:categorical}, this would be a rather
empty notion without some constraints on $\qcal(g)$ [for example,
$\qcal(g)=-\wickprod{\rho}(|g|^2)$ gives a trivial inequality of this type].
To qualify as a nontrivial inequality, $\qcal$ should be of `lower order' than
$\wickprod{\rho}$ in a defined sense. For example, the nonminimally
coupled scalar field  obeys bounds of the form 
\begin{equation} \label{eqn:nonMinCoupling}
\wickprod{\rho}(|g|^2) \ge -\qcal_1(g)\idop +2 \xi \wickprod{\phi^2}(\dot{g}^2)
\end{equation}
in four-dimensional Minkowski space for coupling $\xi\in[0,1/4]$ \cite{FewOst:qei_nonmin}.
Crucially, the right-hand side is bounded
relative to $(1+H)^p$ for any $p>2$, while the left-hand side is not
bounded relative to any $(1+H)^q$ with $q<3$, where $H$ is the
Hamiltonian. 

In the present paper we will weaken the criterion of nontriviality
slightly owing to the approximate nature of OPEs. As we explain in
outline in Sec.~\ref{sec:overview} and in detail in Sec.~\ref{sec:ineq}, we
permit bounds containing a remainder term that is of higher order in energetic terms
than the field of interest, but which is vanishing in the small
distance limit. 

All the results mentioned so far rely on positivity of an
underlying classical expression, namely, a sum of squares of fields and
their derivatives; and this is also the focus of the present work. However, it
is important to recall that the energy density of a Dirac field is not expressed in this way; accordingly different
techniques are required to obtain quantum energy inequalities in this case (see
\cite{FewVer:qei_curved,FewMis:flat_Dirac,DawFew:curved_Dirac,Smi:Dirac_abs} 
for spin-$1/2$ and \cite{YuWu:massive_RaritaSchwinger,HLZ:massless_RaritaSchwinger} for spin-$3/2$).

\subsection{Perturbative versus nonperturbative approaches} \label{sec:power_series}

While QIs were first studied for free fields on Minkowski
space, it is now known -- as mentioned above -- that the concept is
compatible at least with some simple types of interaction, specifically the
coupling to an external gravitational field and those in conformal field theories.
However, on the technical side, the
existing results typically rely on the rather simple structure of linear quantum fields
fulfilling $c$-number commutation relations. For dealing with
general, possibly self-interacting quantum fields, this is far too
restrictive. Instead, our aim here is to derive inequalities from general principles of
quantum field theory that are not restricted to linear fields. 

To date, self-interacting quantum field theories have generally been
established in a perturbative setting only, usually without any
control on the convergence of the perturbation series. It would seem natural to
investigate QIs in this context. However, severe
conceptual difficulties arise here. In order to give any reasonable meaning to
quantum inequalities in perturbation theory, we need to determine when a formal power series, say
$P[g]=\sum_{k=0}^\infty c_k g^k$ with $c_k\in\cbb$, and with the formal
variable $g$ being interpreted as a ``coupling constant'', should be considered
positive. Understanding the set of formal power series as a 
$\ast$-algebra, the natural notion of positivity is as follows
\cite{DueFre:qed_observables}: $P$ is considered positive if and only if
\begin{equation}  \label{eqn:fpsPosSquare}
  P[g] = Q\st[g] Q[g] \quad\text{for some formal power series }Q. 
\end{equation}
It turns out that this condition is equivalent to the following one:
\begin{equation} \label{eqn:fpsPosFirst}
  P[g] = g^{2n} \sum_{k=0}^\infty d_k g^k\quad
  \text{with } n \in \nbb_0, \, d_k \in \rbb, \, d_0 > 0.
\end{equation}
[Here \eqref{eqn:fpsPosSquare} $\Rightarrow$ \eqref{eqn:fpsPosFirst} is
immediate; the converse follows by inserting $x=(d_0^{-1} g^{-2n} P[g]-1)$ into
the power series of $\sqrt{1+x}$ around $x=0$.] Now
Eq.~\eqref{eqn:fpsPosFirst} shows that this notion of positivity is not useful in our context, since it
roughly says that positivity of $P$ is determined by its lowest-order
coefficient. (See \cite{BorWal:formal_gns} for a slight variant.)  The order-0
coefficient however is supposed to be the contribution from free field theory.
So---with this definition---QIs would hold at finite coupling if
and only if they hold at coupling $g=0$; the effects of interaction on
inequalities cannot be captured in this approach.

Let us illustrate these difficulties in a simple example: Should one consider
the following formal power series positive?
\begin{equation}
   P[g] = \sum_{k=0}^{\infty} \frac{(-1)^k}{(2k)!} g^{2k}
\end{equation}
Considering $P$ as a convergent series, it would be positive for small, but not
for all $g$. Forgetting all convergence properties, the only information that
remains is positivity at $g=0$, i.e., of the zero-order coefficient. The
question of interest, however, would be whether the \emph{physical} value of
$g$ falls into the convergence radius of $\sqrt{P}$; this question is
not accessible in formal perturbation theory.

It is therefore necessary to conduct our investigation in a nonperturbative
formulation of quantum field theory, such as the Wightman setting
\cite{StrWig:PCT} or the $C\st$ algebraic formulation of Haag and Kastler
\cite{Haa:LQP}. (We shall actually use a combination of both; the technical
details will be recalled in Sec.~\ref{sec:outset}.) This is, in a way, a very
strong assumption to start with, since we assume that our QFT models have been
fully constructed and are under complete topological control. Indeed, the
rigorous construction of interacting models in physical space-time still remains
an open challenge, while the situation is better in simplified low-dimensional
models \cite{GliJaf:quantum_physics}. The virtues of our axiomatic approach,
however, are of a different nature: Within the framework of algebraic quantum
field theory, we can formulate physically motivated, qualitative properties of
quantum field theories, which can explicitly be verified in simple models such
as free field theory, but which appear general enough to be postulated for the
interacting situation. We can then show how observable consequences, such as
quantum inequalities, follow from these postulated properties.

\subsection{Phase space conditions}\label{sec:phase_space}

The specific qualitative properties we will employ are known as
phase space conditions. Semi-classical considerations
(originating with Bohr and Sommerfeld)
suggest that only finitely many independent states (or, dually,
observables) are required to describe a quantum system which is restricted to a finite volume in
phase space e.g., by cut-offs in configuration space and energy. 
In quantum field theory, this picture can certainly persist only
qualitatively and in an approximate sense. However, it is possible to give a
precise meaning to the aforementioned concepts, expressed as the compactness or
nuclearity of certain maps; see
e.g.~\cite{HaaSwi:compactness,BucWic:causal_independence,BucPor:phase_space}. 
These phase space conditions have physically interesting consequences: for
example, they imply the existence of thermal equilibrium states
\cite{BucJun:equilibrum} and are important for the particle interpretation of
quantum field theories \cite{Por:pw_disint_ii}.
 
The role of phase space conditions for QIs has been partially investigated before. Even in
the free field situation described above, one may see the need for some
restrictions on the phase space behaviour of the theory \cite{Few:Bros}:
instead of one field of mass $m$, consider an infinite number of fields with
masses $m_j$ (for simplicity, in four-dimensional Minkowski space). The total
energy density will obey a QI
\begin{equation}
\int dt\, \omega(\wickprod{\rho}(t,0)) |g(t)|^2 \ge -\frac{1}{16\pi^3}\int_0^\infty
du\, u^4 N(u) |\tilde{g}(u)|^2,
\end{equation}
where 
\begin{equation}
N(u) =\sum_{j} \vartheta(u-m_j)
\end{equation}
counts the number of species with masses below energy $u$. 
If $N$ grows no faster than polynomially with $u$, the lower bound is
finite for all $g\in C_0^\infty(\rbb)$; the same condition is known to 
guarantee that this theory obeys nuclearity in the sense of
Buchholz and Wichmann \cite{BucWic:causal_independence}. Other ideas
concerning the relationship between QEIs and nuclearity conditions are
discussed in \cite{FOP:p_nuclearity}, while connections with
thermodynamic stability are described in \cite{FewVer:stability} and
\cite{SchVer:LTE_QEI}.
However, the results presented here are the first in which QIs have been
derived as a consequence of phase space criteria. 

For our purposes, we will use a {\em microscopic phase space condition}
recently introduced by one of us (HB) in \cite{Bos:short_distance_structure}; we
shall recall its formulation and consequences in Sec.~\ref{sec:outset}.
Compared with other similar conditions, it is specifically sensitive in the
short-distance regime, the realm which is of most interest for QIs. Indeed, one
heuristically expects \cite{HaaOji:germs} that at short distances and
finite energies, the theory may be well-approximated in terms of finitely
many observables corresponding to pointlike quantum fields.

This approximation of bounded observables by
quantum fields can indeed be made precise \cite{Bos:short_distance_structure}
and plays a central role in our approach. Its use is twofold. First, it tells us
how our primary objects---local algebras of bounded operators---relate to
the quantum fields for which inequalities are formulated. Second, it serves to
establish an additional structure for the quantum fields, namely a rigorous
version of the operator product expansion \cite{Bos:product_expansions}.
We can understand this OPE, which describes the ``structure constants" of the
``improper algebra" of quantum fields, as containing all relevant information
about the interaction, and in this sense as a replacement for the Lagrangian
\cite{Wil:non-lagrangian}. In fact, it is the OPE from which our inequalities
will be computed. In particular, the OPE allows us to generalize the notion
of normal ordering that has a key role for QIs of linear
fields, replacing it with normal products in the sense of Zimmermann \cite{Zim:Brandeis}.

The remainder of the paper is organized as follows. We start with a
non-technical account of our main methods and results in
Sec.~\ref{sec:overview}. Then, in Sec.~\ref{sec:outset}, we introduce the framework of nonperturbative
quantum field theory that we work in, including the phase space condition mentioned above. Section~\ref{sec:distributions} presents some technical preliminaries from distribution theory. In Sec.~\ref{sec:products}, we establish the rigorous operator product expansion in the variant that we
require. This expansion will be the base of our quantum inequalities, derived
in Sec.~\ref{sec:ineq}. Dilation covariance as a special case is covered 
in Sec.~\ref{sec:scaling}. We end with a brief outlook in
Sec.~\ref{sec:outlook}.

\section{Overview} \label{sec:overview}

We now give a non-technical overview of our main techniques
and results, postponing rigorous arguments to later sections.
Throughout, we work in Minkowski space of dimension $2+1$ or more (possible
generalizations are discussed in Sec.~\ref{sec:outlook}). For simplicity, we shall always pick a
fixed Lorentz frame, and hence a fixed time axis; all quantum fields
$\phi(t) = \phi(t,\vec x=0)$ will be restricted to this time axis, 
and smeared expressions $\phi(f)=\int dt \, f(t) \, \phi(t)$ will refer to
one-dimensional integration only. This is
sufficient for regularizing Wightman fields \cite{Bor:spacelike_smooth};
due to the symmetry properties of Minkowski space, 
it covers the essential features of the inequalities we wish to consider.

To illustrate our approach we begin by sketching the derivation of a QI for
the Wick square of the free real scalar field,
essentially following the argument of \cite{Few:general_wordline}
but in a form which is amenable to our generalization. We will then
indicate which changes are necessary to deal with the general
situation.

Accordingly, let $\phi$ denote the free field and let $\sigma$ be a 
normal state in the vacuum sector with sufficiently regular high-energy behaviour 
that the expectation values in the following are finite. 
The distributional integration kernel
\begin{equation}
   F(t,t') = \sigma( \phi(t) \phi(t') ).
\end{equation}
is positive-definite, in the sense that for any test function
$g$, we have
\begin{equation} \label{eqn:fPos}
   \int dt\,dt'\, F(t,t') \overline{g(t)} g(t') = \sigma( \phi(g)\st \phi(g) )
   \geq 0.
\end{equation}
Then, also $F(t,t') / \imath\pi(t-t'-\imath 0)$ is positive-definite; namely we
have by Fourier analysis,
\begin{equation} \label{eqn:modKernelPos}
\begin{split}
   \int dt\,dt'\, F(t,t') \frac{ \overline{g(t)} g(t')}{\imath\pi(t-t'-\imath
   0)} = \int \frac{dp}{\pi}  \int dt \, dt'\, F(t,t') \overline{g(t)} g(t')
  \theta(p) e^{-\imath p (t-t')}
\\
  = \int_0^\infty \frac{dp}{\pi} \int dt \, dt'\, F(t,t') \overline{e^{\imath
  p t}g(t)} e^{\imath p t'} g(t') \geq 0.
\end{split}
\end{equation}
We now use Wick ordering and introduce new variables $s=(t+t')/2$, $s'=t-t'$ 
in order to rewrite the kernel $F$:
\begin{equation} \label{wickSplit}
   F(t,t') =  \sigma( \wickprod{\phi^2}(s) ) + \Delta_+(s') + \sigma(R(s,s')),
\end{equation}
where $\Delta_+(t-t') = \omega(\phi(t)\phi(t'))$ is the vacuum two-point function of $\phi$, 
and the remainder $R$ is given by
\begin{equation}
\begin{aligned}
    R(s,s') &=  \; \wickprod{\phi(t) \phi(t')} - \wickprod{
    \phi\big(\frac{t+t'}{2}\big)^2}
  \\
  &=  \; U(s) \big( \wickprod{\phi(s'/2) \phi(-s'/2) - \phi^2(0)} \big)
  U(s)\st ;
\end{aligned}
\end{equation}
it is a smooth function when evaluated in $\sigma$. Inserting this into
Eq.~\eqref{eqn:modKernelPos}, we obtain
\begin{equation} \label{eqn:ineqWithRemainder}
   \sigma(\wickprod{\phi^2}(f) + c_g \idop) \geq - R_{\sigma,g},
\end{equation}
where
\begin{align}
    \label{eqn:fIntegral}
    f(s) &:= \int ds' \frac{ \overline{g(s+s'/2)} g(s-s'/2)}{\imath\pi(s'-\imath 0)},
\\
   \label{eqn:cIntegral}
    c_g &:= \int ds \, ds' \Delta_+(s') \frac{ \overline{g(s+s'/2)} g(s-s'/2)}{\imath\pi(s'-\imath 0)},
\\
   \label{eqn:remIntegral}
   R_{\sigma,g} &:= \int ds\, ds'\, \sigma(R(s,s')) \frac{ \overline{g(s+s'/2)}
   g(s-s'/2)}{\imath\pi(s'-\imath 0)}.
\end{align}
It seems plausible that $R_{\sigma,g}$ becomes small as $\supp g$ shrinks to a
point. We will give more quantitative estimates in that respect later. Here,
let us consider the special case where $g$ is real-valued. Then 
both $\overline{g(s+s'/2)} g(s-s'/2)$ and $R(s,s')$ are even functions in $s'$.
Hence, in Eqs.~\eqref{eqn:fIntegral} and \eqref{eqn:remIntegral}, we can replace the factor $(s'-\imath 0)^{-1}$ with its even part,
\begin{equation}
   \frac{1}{2} \Big(  \frac{1}{s'-\imath 0} +\frac{-1}{s'+\imath 0} \Big) = \imath \pi \delta(s').
\end{equation}
Since $R(s,0)=0$, this results in $R_{\sigma,g}=0$ and $f(s)=g(s)^2$. Thus
Eq.~\eqref{eqn:ineqWithRemainder} gives the more usual inequality for the Wick
square,
\begin{equation} \label{eqn:usualQi}
   \wickprod{\phi^2} (g^2) \geq - c_g \idop.
\end{equation}

We now aim at a generalization beyond free field theory. So let $\phi$ be a
general, possibly self-interacting local quantum field. 
(The term ``quantum field" is used here in a generic
fashion, and may include derivatives of fields as well as composite fields or
suitably defined powers of fields.) The main
difficulty we face in applying the above construction is that no concept of
normal ordering is available; we cannot use Wick ordering to split the
product into higher-order and lower-order terms, as in Eq.~\eqref{wickSplit}.
Instead, we shall use an operator product expansion for the product
$\phi(t)\st\phi(t')$,
\begin{equation} \label{eqn:squareOpe}
\phi^*(t)\phi(t') = \sum_{j=1}^n C_j(t-t') \phi_j((t+t')/2) + R_n(t,t').
\end{equation}
Here $R_n$ is a remainder term, which is ``small'' where $t$ and $t'$
are close, while the $\phi_j$ are composite
fields. 
Smearing against $\overline{g(t)}g(t')$, where $g\in
\dcal(\rbb)$, the left-hand side is then a positive operator, 
and this remains true if we multiply \eqref{eqn:squareOpe}
with any positive-type kernel $K(t-t')$, which takes the role of
$K(t-t')=1/\imath \pi(t-t'-\imath 0)$ above. In other words,
Eqs.~\eqref{eqn:fPos} and \eqref{eqn:modKernelPos} remain valid. We can then rearrange and obtain as
analogue to Eq.~\eqref{eqn:ineqWithRemainder},
\begin{equation}
\sum_{j=1}^n \phi_j(f_j) \ge - \int dt\,dt'\,\overline{g(t)}g(t') K(t-t')
R_n(t,t'),
\end{equation}
where the test functions $f_j$ are given in terms of $g$, $K$, and the OPE
coefficients $C_j$ by
\begin{equation} \label{eqn:fjPreview}
f_j(s) = \int ds'\, K(s') \,C_j(s')\, \overline{g(s+s'/2)}g(s-s'/2).
\end{equation}
Note that there is no guarantee that these functions are
necessarily pointwise positive (the issues here are 
related to Hudson's theorem~\cite{Hud:Wigner} and the `choice of basis'
invoked in the OPE). We will return to this below. 

{}In order to establish our results rigorously, the main task is to
establish the OPE and to control the remainder term on the right-hand side. 
We will show in Theorem~\ref{thm:qi} that, given $\alpha\ge 0$, one may find
$n$, $m$ and $\ell$ so that for all $d>0$ and $g\in\dcal(-d,d)$,
\begin{equation} \label{eqn:qiPreview}
\sum_{j=1}^n  \phi_j(f_j) \ge - \epsilon(d)
\|g\|_{d,m}^2 (1+H)^{2\ell},
\end{equation}
where $H$ is the Hamiltonian, $\epsilon(d)=o(d^\alpha)$ as $d\to 0$ and $\|\cdot\|_{d,m}$ is
equivalent to the Sobolev norm on $W_0^{m,1}(-d,d)$. (Of course, finite
sums of field products can be, and are, accommodated by our result.)

The relationship with the QIs described in Sec.~\ref{sec:QI_background} is most apparent in the
case where one of the composite fields, say $\phi_1$, is of higher order
than the others, in the sense that there exists $\ell'$ for which
$(1+H)^{-\ell'}\phi_j(f)(1+H)^{-\ell'}$ is bounded for $j\ge 2$, while 
$(1+H)^{-\ell'}\phi_1(f)(1+H)^{-\ell'}$ is unbounded. Then we may
rearrange to write
\begin{equation}\label{eq:rearrangedQI}
\phi_1(f_1) \ge - \sum_{j=2}^n \phi_j(f_j)  - \epsilon(d)
\|g\|_{d,m}^2 (1+H)^{2\ell}.
\end{equation}
In cases where the remainder term vanishes \eqref{eq:rearrangedQI} is
then a nontrivial QI in the sense of \cite{Few:categorical}: namely,
one cannot find constants $C, C'$ such that
\begin{equation}
|\sigma(\phi_1(f_1))| \le C\sum_{j=2}^n |\sigma(\phi_j(f_j))|+ C'
\end{equation}
for all (sufficiently regular) states $\sigma$ because $\phi_1$ is of higher
energetic order than the fields on the right-hand side. Examples include
the QI~\eqref{eqn:usualQi}, where the only
composite field on the right-hand side is the identity, and the QEI
\eqref{eqn:nonMinCoupling} on the nonminimally coupled field,
where both the identity and Wick square appear on the right-hand side. 

This simple situation does not persist in general, however. First, 
it does not seem guaranteed that a unique choice of a highest-order
field $\phi_1$ exists. 
For interacting fields, one would expect $\phi_1$ to be the normal product 
of $\phi\st\phi$ in the sense of Zimmermann \cite{Zim:Brandeis};
but there are indications from perturbation theory that in some cases, this
normal product might not be unique \cite{Joh:green_2d}. Second, the
remainder term cannot be expected to vanish in general--this reflects
that the OPE is a controlled approximation, rather than an exact
formula. Third, the remainder term is not of lower energetic order than the
fields: in fact, $\ell$ is chosen so that each
$(1+H)^{-\ell}\phi_j(f_j)(1+H)^{-\ell}$ is bounded.

Although the 
inequalities in Eq.~\eqref{eqn:qiPreview} remain valid, it is necessary
to adapt the criterion of nontriviality to our setting. Our approach is
to focus on the short-distance behaviour, in which the remainder term
vanishes as $o(d^\alpha)$. By contrast,  
we will show in Sec.~\ref{sec:nontriv} that, 
for the bounds we obtain,
\begin{equation} 
\sup_{g\in\dcal(-d,d)} \|g\|_{d,m}^{-2}\,\Big\| (1+H)^{-\ell} \sum_{j=1}^n 
\phi_j(f_j) (1+H)^{-\ell}\Big\|
\end{equation}
is not $o(d^\alpha)$ as $d\to 0$. Thus the remainder term cannot
dominate the contribution of the composite fields in the small. 
In this context, it
turns out to be crucial to formulate the OPE, and correspondingly the
inequalities, in a ``basis-independent'' fashion, that is, in a way that is
independent of a possible arbitrariness in the choice of composite fields. 

For practical purposes, it is still important to understand the more
specific question as to whether there is a normal product of strictly higher
energetic order than the other fields in the OPE. At present it seems to
us that this must be discussed in the light of particular examples. 

Last but not least, one would like to gain more insight in the properties of
the sampling functions $f_j$, given by Eq.~\eqref{eqn:fjPreview}, in particular for
the function $f_1$ corresponding to a ``highest-order'' composite field,
which generalizes $f_1(s)=g(s)^2$ from the free field case. In general, it is
certainly not expected that $f_1$ depends on $g$ in a simple pointwise fashion.
However, one may ask whether $\phi_1$ can be chosen so that $f_1$ retains other
properties that are apparent in the free-field situation, for example whether
$f_1 \geq 0$, either pointwise or in an averaged sense. This will depend
crucially on the form of the OPE coefficients $C_j$, which are however unknown
in general. We will give two approaches to this problem. The first, in
Sec.~\ref{sec:mesoscopic}, indicates conditions under which one
may simultaneously tune the leading sampling function to a given
positive form, while also reducing the remainder term. These conditions are
broadly met under the assumption that the OPE coefficients have scaling
limits in the sense of \cite{FreHaa:covariant_scaling}. Our argument
here is essentially to form a Riemann sum of QIs over small distance
scales, in which the remainder term is suppressed. This approach is,
however, tied to basis representations of the OPE; it would appear to be
most useful in the context of particular models. Second, in
Sec.~\ref{sec:scaling}, we discuss the particular case of dilation covariant theories.
This is of interest since we can expect that our theory is approximated by a dilation
covariant ``scaling limit theory'' in the ultraviolet. In this restricted
situation, we will derive explicit criteria on $g$ that guarantee positivity of
$f_1$. Since positivity of $f_1$ also fixes the sign of the composite field
$\phi_1$, this gives us a means of distinguishing the positive ``normal
square'' of $\phi$ from its negative.

\section{Algebras of observables and pointlike quantum fields}
\label{sec:outset}

As a mathematical basis of quantum field theory,
we adopt the framework of local quantum physics \cite{Haa:LQP}.
Specifically, for describing pointlike quantum fields, we use the
methods set forth in \cite{Bos:short_distance_structure}. For the convenience
of the reader, we will collect the relevant notions and results below, and
introduce some notations that are useful in our context.

We set out from a local net of algebras, $\ocal \mapsto \afk(\ocal)$, in the
vacuum sector. That is, for each bounded open region $\ocal$ of Minkowski space, 
we have an algebra $\afk(\ocal)$ of bounded operators; we take these to
be von Neumann algebras acting on a common Hilbert space $\hcal$. Further, we
have a strongly continuous unitary representation $(x,\Lambda) \mapsto
U(x,\Lambda)$ of the proper orthochronous Poincaré group
on $\hcal$, with a common invariant unit vector $\Omega\in\hcal$. We write the
translation subgroup as $U(x,\idop)=\exp\imath P_\mu x^\mu$. Together these objects are supposed to fulfil the following axioms:

\begin{enumerate}
\localitemlabels
  \item \emph{Isotony:} $\afk(\ocal_1) \subset \afk(\ocal_2)$ if
  $\ocal_1 \subset \ocal_2$.
  \item \emph{Locality:} $[A_1, A_2] = 0$ if $\ocal_1,\ocal_2$ are two
  spacelike separated regions, and $A_i \in \afk(\ocal_i)$.
  \item \emph{Covariance:} $U(x,\Lambda) \afk(\ocal) U(x,\Lambda)\st =
  \afk(\Lambda \ocal + x)$ for all Poincaré transformations $(x,\Lambda)$.
  \item \emph{Positivity of energy:} The joint spectrum of the
  $P_\mu$ falls into the closed forward light cone.
  \item \emph{Uniqueness of the vacuum:} $\Omega$
  is unique (up to a phase) as an
  invariant vector for all $U(x,\idop)$.
\end{enumerate}

We are primarily interested in the algebras associated with standard double cones
$\ocal_r$ of radius $r$ centred at the origin, and use $\afk(r)$ as shorthand
for $\afk(\ocal_r)$. Also, for most parts we only use the time-translation
subgroup of $U(x,\Lambda)$, which we denote as 
$t\mapsto U(t)$, with positive generator $H=P_0\geq 0$. We write the
spectral projectors of $H$ for the interval $[0,E]$ as $P(E)$.

Let $\Sigma$ be the set of ultraweakly continuous
functionals on $\boundedops$.  We consider for $\ell>0$ the 
subspaces 
\begin{equation}
\sigspace{\ell} = \bigsetprop{ \sigma \in \Sigma }{ \lnorm{\sigma}{\ell} :=
\|\sigma \big((1+H)^{\ell} \cdotarg (1+H)^{\ell} \big)\| < \infty },
\end{equation} 
which are Banach spaces in the norm $\lnorm{\cdotarg}{\ell}$. Their duals
$\sigspace{\ell}\st$ consist of linear forms $\phi$ for which the dual norm
$\lnorm{\phi}{-\ell} = \| (1+H)^{-\ell} \phi (1+H)^{-\ell}\|$ is finite. [More
precisely, $\phi$ are quadratic forms on a dense subspace of
$\hcal\times\hcal$, for which the form $(1+H)^{-\ell} \phi (1+H)^{-\ell}$,
with the multiplication defined in the weak sense, is bounded.]

We also introduce the space of smooth functionals, 
$\cinftys = \cap_{\ell>0} \sigspace{\ell}$, and
equip it with the Fr\'echet topology induced by all norms
$\lnorm{\cdotarg}{\ell}$. The dual space $\cinftyss$ is then given by
$\cup_{\ell>0} \sigspace{\ell}\st$, and will be considered with the weak$\ast$
topology. Further, we define for $E>0$ the set of energy-bounded functionals,
$\Sigma(E)=\{\sigma(P(E)\cdotarg P(E))\,|\, \sigma \in \Sigma\}$. Then
$\cup_{E>0}\Sigma(E)$ is dense in $\cinftys$ and weakly dense in $\Sigma$. Each
space $\sigspace{\ell}$ is invariant under the natural action of hermitean
conjugation, i.e.~$\sigma\st(A) = \overline{\sigma(A\st)}$, and this structure
transfers to the dual spaces; so we can speak of hermitean elements in
$\cinftys$ and $\cinftyss$.

With respect to pointlike fields, we assume that the theory fulfils a specific 
type of phase space condition \cite{Bos:short_distance_structure}, sensitive in
the ultraviolet. To formulate this, consider the inclusion map $\Xi:
\cinftys \hookrightarrow \Sigma$. We assume that $\Xi$ can be approximated with finite-rank maps in the following sense. 

\begin{definition}[Microscopic phase space
condition]\label{def:microPhase} A net $\ocal \mapsto \afk(\ocal)$ is said to satisfy the
  {\em microscopic phase space condition} if for every $\gamma \geq 0$,
  there exists a linear continuous map $\lcssample: \cinftys \to \Sigma$ of
  finite rank such that for sufficiently large $\ell>0$,
  \begin{equation*}
    r^{-\gamma} \lnorm{ (\psmap-\lcssample) \restrict \afk(r)}{\ell} \to 0
    \quad \text{ as }r \to 0.
  \end{equation*}
\end{definition}

Here the restriction $\restrict \afk(r)$ is applied to the image points of
the maps, which are functionals in $\Sigma$. This phase space condition is known
to be fulfilled in free field theory in at least $3+1$ space-time dimensions, for massive free fields also in $2+1$ dimensions \cite{Bos:operatorprodukte}.

The consequences of this condition are as follows \cite{Bos:short_distance_structure}.
While the maps $\psi$ are not uniquely fixed by the property above,
the image of their dual maps, $\img \psi\st =: \Phi_\gamma$, 
is actually unique at fixed $\gamma$, provided that the rank of $\psi$ is
chosen minimal. These finite-dimensional spaces $\Phi_\gamma$ form an
increasing sequence $\Phi_0 \subset \Phi_1 \subset \Phi_2 \ldots$,
and their union $\cup_\gamma \Phi_\gamma = \PhiFH$ is precisely the field
content of the theory as defined by Fredenhagen and
Hertel~\cite{FreHer:pointlike_fields}. After smearing with test functions, 
the elements $\phi \in \Phi_\gamma$ are local Wightman fields. 
Actually it suffices for regularizing $\phi$ to smear it along the time axis; 
that is, for $f \in \scal(\rbb)$ and $\phi \in \PhiFH$, 
the quadratic form $\phi(f) = \int dt\,f(t)\,U(t) \,\phi \,U(t)\st$ can be 
continued to an unbounded, but closable operator on the dense invariant domain
$\cinftyh = \cap_{\ell>0} (1+H)^{-\ell} \hcal$. Further, $\phi \in \Phi_\gamma$
can be approximated with bounded operators in a controlled way; 
cf.~\cite[Lemma~3.5]{Bos:short_distance_structure} and the remark following it:
\begin{theorem} \label{thm:fieldConverge}
Let $\phi \in \PhiFH$. One can find constants $\ell>0$, $k>0$ 
and operators $A_r \in \afk(r)$ for each $r>0$ such that, as $r \to 0$,
\begin{align*}
  &\lnorm{\phi}{-\ell} < \infty,
   \qquad
  &\lnorm{A_r - \phi}{-\ell} = O(r),
  \qquad
  \lnorm{A_r}{-\ell} = O(1),
  \\
  &\|A_r\| = O(r^{-k}),
  & \forall n \in\nbb:\, \|\frac{d^n}{dt^n} U(t) A_r U(t)\st \| = O(r^{-k-n}).
\end{align*}
\end{theorem}

Moreover, the spaces $\Phi_\gamma$ are related to the approximation 
of bounded operators in the short distance limit; see \cite[Eq.~(4.4)]{Bos:short_distance_structure}:
\begin{theorem} \label{thm:projectorApprox}
Let $p_\gamma:\cinftyss\to\Phi_\gamma \subset \cinftyss$ be a continuous
projection onto $\Phi_\gamma$. Then, for sufficiently large $\ell>0$, 
\begin{equation*}
   \lnorm{ (\psmap p_{\gamma\ast} - \psmap) \restrict \afk(r) }{\ell} 
  = o(r^\gamma).
\end{equation*}
\end{theorem}
Here $p_{\gamma\ast}: \cinftys \to \cinftys$ is the pre-dual map to $p_\gamma$,
which always exists due to its finite rank. Of course, such projections
$p_\gamma$ exist in abundance. Since the spaces $\Phi_\gamma$ are invariant
under conjugation, it is possible to choose $p_\gamma$ hermitean, i.e., such
that $p_\gamma(A\st)=p_\gamma(A)\st$.

It was shown in \cite{Bos:product_expansions} that due to the properties explained above, operator
product expansions exist in a rigorous sense. In fact, \cite{Bos:product_expansions}
established the expansion of $\phi(x)\phi'(y)$ for spacelike separated points $x$ and $y$.
A similar scheme can be applied for arbitrary $x$ and $y$, in the sense of distributions,
as sketched in \cite{Bos:product_expansions} and worked out in more detail
in \cite[Ch.~5.5]{Bos:operatorprodukte}. (See also \cite[Sec.~4]{BDM:field_renorm}.) For our purposes, we will need a specific variant of
this product expansion, which will be established in Sec.~\ref{sec:ope}.

\section{Distributions as boundary values of analytic functions}
\label{sec:distributions}

If $\sigma \in \Sigma$ is energy-bounded and $\phi$ a Wightman field with
sufficiently regular high-energy behaviour, then the distribution\footnote{%
Throughout the paper, we will write 
distributions in terms of their formal integration kernels, such as $\int K(x)
f(x) dx$ for the evaluation of a distribution $K \in \scal(\rbb)'$ on a test
function $f \in \scal(\rbb)$, even if $K$ does not arise from an integrable
function or measure. This is merely a notational convention.} 
$\sigma(\phi\st(t) \phi(t'))$ is the boundary value of an analytic function in
the half plane $\im (t-t')<0$. If further $\sigma$ is positive, then the
distribution is positive-definite. These types of distributions have certain
well-known characterizations \cite{StrWig:PCT,ReeSim:mmmp2}. Since we will need
specific quantitative estimates in our context, we will repeat some of those
arguments in detail.

First of all, for each $d>0$ and $m \in \nbb$ we define a norm on $\dcal(-d,d)$
by
\begin{equation}
  \sobnorm{f}{d}{m} := \max_{0 \leq n \leq m} d^n \|f^{(n)}\|_1.
\end{equation}
This norm is equivalent to the Sobolev norm defining the space
$W_0^{m,1}(-d,d)$ \cite{AdaFou:sobolev}, but it is convenient to use the above
norms owing to their behaviour under scaling. Namely, if $f \in \dcal(-d,d)$ and $\lambda >0$, and we set
$f_\lambda(t) = \lambda^{-1}f(t/\lambda)$, then $f_\lambda \in \dcal(-\lambda
d, \lambda d)$ and
\begin{equation}
\forall m \in \nbb: \quad  \sobnorm{f_\lambda}{\lambda d}{m}
=  \sobnorm{f}{ d}{m}.
\end{equation}

Let us now define the class of analytic functions that is of
interest.

\begin{definition} \label{def:boundary}
  We say that an analytic function $F: \rbb-\imath\rbb_+ \to \cbb$ is
  \emph{regular at the boundary} if there exists $\ell>0$ such that
  \begin{equation*}
  \lnorm{F}{-\ell}:=\sup_{-1 \leq \im z < 0} |F(z)|\,|\im z|^{\ell}
  \end{equation*}
  is finite. The space of all such functions for given $\ell$ is denoted
   as $\kcal^\ell$; and $\kcal := \cup_{\ell>0} \kcal^\ell$.
\end{definition}

As the name suggests, functions in $\kcal$ have
distributional boundary values on the real line.

\begin{proposition} \label{pro:analyticBoundary}
Let $F \in \kcal^\ell$. Then the limit $ \lim_{y \to 0+} F(x-\imath
y)$ exists as a tempered distribution in $x$. 
The limit distribution $F(x-\imath 0)$ satisfies the following estimate
for $f \in \dcal(-d,d)$, $d>0$: 
\begin{equation*}
   \big|\int f(x) \, F(x-\imath 0)\, dx \big| \leq 
   4^{\ell+2} \; (\ell+3)  (1+d^{-\ell-2}) \;
    \lnorm{F}{-\ell} \; \sobnorm{f}{d}{\ell+2}.
\end{equation*}
\end{proposition}

\begin{proof}
For fixed $f \in \scal(\rbb)$, consider the function 
\begin{equation} \label{gDef}
  g(y) := \int  f(x) F(x-\imath y)dx, \quad
  0 < y \leq 1.
\end{equation}
Since $F$ is analytic in $z = x-\imath y$, we can obtain the derivatives of $g$ using integration by parts:
\begin{equation}
\forall j \in \nbb_0: \quad  \frac{d^j g}{dy^j} (y) = \int  f(x)
(-\imath)^j \frac{d^j }{dz^j} F(x-\imath y) dx = \imath^j \int   f^{(j)}(x)  F(x-\imath y) dx.
\end{equation}
Thus we have the estimate
\begin{equation} \label{gDerivEst}
\forall j \in \nbb_0:\quad   \big|\frac{d^j g}{dy^j} (y) \big|
\leq \, y^{-\ell} \lnorm{F}{-\ell}\, \|f^{(j)}\|_1 .
\end{equation}
We now want to deduce the following improved estimate.
\begin{equation} \label{gDerivEst2}
  \big|\frac{d^j g}{dy^j} (y) \big| \; \leq 
	4^{\ell-j+2} \;   (1+y^{3/2-j})\; \lnorm{F}{-\ell} \sum_{k=0}^{\ell+2}
	\|f^{(k)}\|_1 \quad \text{for $j \in \{0,\ldots, \ell+2\}$}.
\end{equation}
In fact, for $j = \ell+2$, this directly follows from Eq.~\eqref{gDerivEst}.
Now suppose that Eq.~\eqref{gDerivEst2} holds for $j+1$ in place of $j$. We
compute:
\begin{multline} 
  \Big|\frac{d^j g}{dy^j}  (y) \Big| \; \leq 
\Big|\frac{d^j g}{dy^j}  (1) \Big| + \int_{y}^{1} dy' \Big|\frac{d^{j+1} g}{dy^{j+1}}  (y') \Big|
\\
 \leq \lnorm{F}{-\ell}\,\|f^{(j)}\|_1 
+  4^{\ell-j+1} \; \lnorm{F}{-\ell} \;  \sum_{k=0}^{\ell+2} \|f^{(k)}\|_1 
\int_y^1 dy' (1+(y')^{1/2-j})
\\
\leq 4^{\ell-j+1} \; \lnorm{F}{-\ell} \sum_{k=0}^{\ell+2} \|f^{(k)}\|_1 ( 
4 + 2 y^{3/2-j} ). 
\end{multline}
This proves Eq.~\eqref{gDerivEst2}. In particular, the case $j=1$ shows that 
$dg/dy$ is bounded as $y\to 0$; thus $g(y)$ converges in this limit. 
Setting $j=0$ in Eq.~\eqref{gDerivEst2} then shows that
$g(0+)=: \int f(x) F(x-\imath 0) dx$ defines a tempered distribution. Also, if
$f \in \dcal(-d,d)$, we can combine the estimate
\begin{equation} \label{eqn:LOneToSobolev}
  \sum_{k=0}^{m} \|f^{(k)}\|_1 \leq (m+1) \, \max\{1,d^{-m}\} \, 
  \sobnorm{f}{d}{m}
\end{equation}
with Eq.~\eqref{gDerivEst2}, where $j=0$ and $m=\ell+2$, in order to show the
proposed estimate for the limit distribution.
\cmpqed\end{proof}

It is clear from Definition~\ref{def:boundary} that, for two functions which
are regular at the boundary, their product inherits this property. 
More explicitly, for $F \in \kcal^\ell$ and $G \in
\kcal^m$, we have
\begin{equation} \label{eqn:regProdrule}
   \lnorm{F G}{-\ell-m} \leq \lnorm{F}{-\ell} \lnorm{G}{-m}.
\end{equation}
Thus the product of the
boundary distributions is well-defined by multiplying the
analytic functions.
On the other hand, the Fourier transforms of the boundary
distributions have support in $[0,\infty)$. This allows for an alternative definition of the distribution
product by convolution in Fourier space. 
The two definitions are in fact equivalent \cite[Ch.~IX.10, Example~4]{ReeSim:mmmp2}.

Apart from our distributions being boundary values of analytic functions, we
also need to consider questions of positivity. We remind the reader of the
definitions (the terminology is not completely consistent in the literature).
For $g_1,g_2 \in \dcal(\rbb)$, we
introduce the abbreviation $g_1 \syp g_2(s,s') := g_1(s+s'/2) g_2(s-s'/2)$.
\begin{definition}
A distribution $K \in \scal(\rbb^2)'$ is called
\emph{positive-definite} if for all $g \in \scal(\rbb)$, one has
$\int ds \,ds'\, K(s,s') \,\bar g \syp g (s,s') \geq 0$. 
If here $K$ depends on the second variable only, so $K \in \scal(\rbb)'$, it is
called a \emph{distribution of positive type}. 
With $\kcal_+ \subset \kcal$ we denote the subset of positive
type distributions.
\end{definition}
The Bochner-Schwartz Theorem asserts that distributions of positive type are
precisely the Fourier transforms of positive, polynomially bounded measures. We
now show that the product of distributions, as discussed further above,
preserves positivity if both factors are positive, at least in a special
situation that is of interest to us.

\begin{proposition} \label{pro:posMultiply}
Let $F \in \kcal_+$. Let $G: \rbb \times (\rbb - \imath \rbb_+) \to \cbb$
such that $G(s,\cdotarg) \in \kcal^\ell$ for some $\ell$ and every fixed $s$,
where the map $\rbb \to \kcal^\ell$, $s \mapsto G(s,\cdotarg)$ is
bounded and continuous in $\lnorm{\cdotarg}{-\ell}$ . 
Suppose further that $G(s,s'-\imath 0)$ is positive-definite.
Then the product distribution $P(s,s')=F(s'-\imath
0)G(s,s'-\imath 0)$ is continuous in $s$ and positive-definite.
\end{proposition}

\begin{proof}
First, due to Prop.~\ref{pro:analyticBoundary}, the boundedness and continuity
of $s \mapsto G(s,\cdotarg)$ implies that $\int P(s,s') f(s') ds'$ is
continuous and bounded in $s$; in particular $P \in \scal(\rbb^2)'$ is
well-defined. Now let $\mu$ be the positive measure that arises by Fourier transform of $F(x-\imath 0)$. 
Since $F(x-\imath 0)$ is a boundary value, we know $\supp \mu \subset
[0,\infty)$. Therefore we have for $g \in \scal(\rbb)$,
\begin{multline}
  \int ds \, ds' P(s,s') \bar{g} \syp g (s,s')
\\
  = \lim_{\epsilon \to 0+} \int_0^\infty d\mu(p) 
  \underbrace{\int ds\,ds' e^{-\imath p(s'-\imath\epsilon)}
  G(s,s'-\imath \epsilon) \bar g \syp g(s,s')}_{=:I(\epsilon,p)}.
\end{multline}
Supposing for a moment that the integrand $I(\epsilon,p)$ has an integrable
bound in $p$, uniform in $\epsilon$, we can apply the dominated convergence
theorem and obtain 
\begin{equation}
  \int ds \, ds' P(s,s') \bar g \syp g (s,s')
  = \int_0^\infty d\mu(p) 
  \int ds\,ds' 
  G(s, s'-\imath 0) \bar g_p \syp g_p (s,s'),
\end{equation}
where $g_p(t)= e^{\imath p t} g(t)$. This is clearly non-negative, since $G$ is
positive-definite.

It remains to prove appropriate bounds for $I(\epsilon,p)$. To that end, choose
$n\in\nbb$ so large that $\int d\mu(p) (1+p)^{-n}<\infty$. We use integration by
parts in $s'$ to obtain
\begin{multline} \label{eqn:iRewrite}
  I(\epsilon,p) = (1+p)^{-n} \int ds \,ds'\, (1+p)^n e^{-\imath
  (s'-\imath\epsilon) p} G(s,s'-\imath\epsilon) \bar g \syp g (s,s')
  \\
  = (1+p)^{-n} \int ds \,ds'\, e^{-\imath (s'-\imath\epsilon) p}
  \Big(1-\imath \frac{\partial }{ \partial s'} \Big)^n
  G(s,s'-\imath\epsilon) \bar g \syp g (s,s').
\end{multline}
Via the Leibniz rule, we can distribute the derivatives $\partial/\partial s'$
to $G(s,s')$ and to the test function. Now note that with $G$, also the
derivatives $\partial^k G/\partial z^k$ fulfil polynomial bounds when $\im z
\to 0-$; namely, we can use the Cauchy integral
formula for a circle of radius $|\im z/2|$ around $z$ in order to obtain the
estimate
\begin{equation}
   \Big| \frac{\partial^k G(s,z)}{\partial z^k} \Big| 
   \leq  2^k k! \, \sup_{x}\lnorm{G(x,\cdotarg)}{-\ell} |\im
   z|^{-\ell-k} \quad \text{for } -\frac{1}{2} \leq \im z < 0.
\end{equation}
This implies that $e^{-\imath p z/2}\partial^k/\partial z^k G(s,z/2)$ belongs to
$\kcal^{\ell+k}$ with norm uniform in $s$ and $p$. Applying Proposition~\ref{pro:analyticBoundary},
we can then obtain finite bounds on the integral in \eqref{eqn:iRewrite} as $\epsilon
\to 0$, so
\begin{equation}
 | I(\epsilon,p) |
  \leq c \, (1+p)^{-n}  \quad \text{for small }\epsilon,
\end{equation}
with a constant $c$ depending on $G$ and $g$. This is a bound
of the required form.
\cmpqed\end{proof}

\section{Products} \label{sec:products}

Our next aim is to describe products of quantum fields that are of interest to
us, and derive an operator product expansion for them. Specifically, we are
interested in the products of two quantum fields $\phi,\phi'$, 
displaced to different points $t,t'$ on the time axis; this
product then exists as a distribution in the difference variable $s'=t-t'$. In addition, we wish to
multiply this distribution with a $c$-number distributional kernel in $t-t'$,
and also consider sums of such expressions. The operator product
expansion we use is derived by means of techniques described in
\cite{Bos:product_expansions}; however, we need to generalise the
construction both to include the weighting factors and also to obtain more
detailed estimates on OPEs at timelike-separated points. 
 
We can formally describe the products of interest as elements of the algebraic tensor
product space $\PhiProd := \kcal \otimes \cinftyss \otimes
\cinftyss$. Any element $\Pi\in\PhiProd$ has the form of a finite sum,
\begin{equation}
  \Pi = \sum_j K_j \otimes \phi_j \otimes \phi_j', \quad
  K_j \in \kcal, \; \phi_j,\phi_j' \in \cinftyss.
\end{equation}
For $\ell>0$, we set $\PhiProdL{\ell}= \kcal^{\ell} \otimes \sigspace{\ell} \st
\otimes \sigspace{\ell} \st \subset \PhiProd$; clearly, $\PhiProd =
\cup_{\ell>0} \PhiProdL{\ell}$.
Further we consider the subspace $\PhiProdLoc = \kcal \otimes \PhiFH
\otimes \PhiFH \subset \PhiProd$, the space of products of pointlike fields. To each
product $\Pi \in \PhiProd$, we can associate a distribution $\opd$,
heuristically given by
\begin{equation} \label{eqn:prodHeuristic}
 U\big(\frac{t+t'}{2}\big) \opd(t-t') U\big(\frac{t+t'}{2}\big)\st  = \sum_j 
 K_j(t-t'-\imath 0) \phi_j(t) \phi_j(t').
\end{equation}
We shall first discuss in which sense these product
distributions exist, before deriving an operator product expansion for them, in the case where 
$\phi_j$ and $\phi_j'$ are local fields. Then we will introduce certain
convolutions of these distributions with test functions, generalize the OPE for them, and single out a minimal set of composite fields that will be of use to us.

\subsection{Operator products} \label{sec:productDef}
Before considering our operator products, let us first define the set of
distributions of interest.

\begin{definition} \label{def:cinftyssDist}
A $\cinftyss$\emph{-valued distribution} is a linear map $T: \dcal(\rbb)
\to \cinftyss$ such that there exist constants $\ell>0$ and $m \in
\nbb_0$, and, for each $d>0$, a constant $c_d$, with the property that
\begin{equation*}
  \forall f \in \dcal(-d,d):  \quad
  \lnorm{T(f)}{-\ell} \leq c_d \sobnorm{f}{d}{m}.
\end{equation*}
\end{definition}
Equivalently, we might say that $T \restrict \dcal(-d,d)$ extends to a map
of $W_0^{m,1}(-d,d)$ to $\sigspace{\ell}\st$, with finite norm
$\lsobnorm{T}{-\ell}{d}{m} \leq c_d$. In more standard terms, $T$
might be called a distribution of finite order, but since we will not use other distributions
in this context, we drop the extra qualifier. 
As before, we shall denote these distributions using their formal kernels:
$T(f) = \int dx\, T(x) f(x)$. Their expectation values
$\sigma(T(x))$, for fixed $\sigma \in \sigspace{\ell}$, are then
distributions in $\dcal(\rbb)'$ in the usual sense. We shall call a
$\cinftyss$-valued distribution \emph{skew-hermitean} if $T(x)\st=T(-x)$.

We shall now clarify in which precise sense the
distributions $\opd$ in Eq.~\eqref{eqn:prodHeuristic} exist. 
\begin{proposition} \label{pro:productExist}
  Let $\ell>0$. To each $\Pi = \sum_j K_j \otimes \phi_j \otimes \phi_j' \in
  \PhiProdL{\ell}$, there exists a unique $\cinftyss$-valued distribution
  $\opd$ such that for any $\sigma \in \cup_{E>0} \Sigma(E)$,
  \begin{equation*}
    \sigma(\opd(f)) = \sum_j \int ds' f(s') K_j(s'-\imath 0)
    \sigma\big(\phi_j(s'/2-\imath 0)\phi_j'(-s'/2+\imath 0)\big).
  \end{equation*}
The map $\Pi \mapsto \opd$ is linear. Further, 
there is a constant $c>0$ such that for any $d\leq 1$,
\begin{equation*}
  \lsobnorm{\opd}{-5\ell-1}{d}{3\ell+2} \leq c \,
  d^{-(3\ell + 2)} \sum_j \lnorm{K_j}{-\ell} \lnorm{\phi_j}{-\ell}
  \lnorm{\phi_j'}{-\ell}.
\end{equation*}
\end{proposition}

\begin{proof}
Without loss of generality, we can assume that $\Pi$ is of the form
$\Pi=K\otimes \phi \otimes\phi'$. Let $\sigma\in\Sigma(E)$, where $E>0$ is
fixed for the moment. Then, due to the spectrum condition, the distribution
$\sigma(\phi(s'/2)\phi'(-s'/2))$ is indeed the boundary value of an analytic
function, namely of
\begin{equation}
  F(s'-\imath s'') := \sigma\big( e^{(s''+\imath s') H/2} \,\phi\, 
  e^{-(s''+\imath s')H} \,\phi'\, e^{(s''+\imath s')H/2}\big), \quad
  s''>0.
\end{equation}
This function fulfils the bounds
\begin{multline} \label{eqn:fieldprodBounds}
  |F(s'-\imath s'')| \leq 
  \|\sigma\| e^{Es''} (1+E)^{2\ell} \lnorm{\phi}{-\ell} \lnorm{\phi'}{-\ell}
  \sup_{\lambda>0} e^{-\lambda s''} (1+\lambda)^{2\ell} 
  \\\leq 
  \|\sigma\| \, \lnorm{\phi}{-\ell}\, \lnorm{\phi'}{-\ell}
  (1+E)^{2\ell} e^{(1+E)s''} (2\ell)^{2\ell} (s'')^{-2\ell}.
\end{multline}
So $F$ is regular at the boundary in the sense of
Definition~\ref{def:boundary}. Rescaling its argument, we explicitly have
\begin{equation}
   \lnorm{F(\frac{z}{1+E})}{-2\ell}
   \leq c \|\sigma\| \lnorm{\phi}{-\ell}\lnorm{\phi'}{-\ell} (1+E)^{4\ell},
\end{equation}
where the constant $c$ depends on $\ell$ only. The distributional product
$K(s'-\imath 0)F(s'-\imath 0)$ therefore exists. Rescaling also $K$, and
applying Proposition~\ref{pro:analyticBoundary} and Eq.~\eqref{eqn:regProdrule}, we
obtain for any $g \in \dcal(-d,d)$ and with another constant $c'$,
\begin{multline} \label{eqn:rescaledBound}
 \big| \int ds' g(s') K(\frac{s'-\imath 0}{1+E})F(\frac{s'-\imath
0}{1+E}) \big|
\\
 \leq c'  \|\sigma\| \, \lnorm{K}{-\ell}\, \lnorm{\phi}{-\ell}\,
 \lnorm{\phi'}{-\ell} (1+E)^{5\ell} (1+d^{-3\ell-2})
 \sobnorm{g}{d}{3\ell+2}.
\end{multline}
Now let $f \in \dcal(-d,d)$, $d\leq 1$. We set $g(s') = \epe^{-1} f(s'/\epe)
\in \dcal(-\epe d, \epe d)$ and obtain using Eq.~\eqref{eqn:rescaledBound} with
$\epe d$ in place of $d$,
\begin{multline}
 \big| \int ds' f(s') K(s'-\imath 0)F(s'-\imath 0)
   \big|
   \\
 \leq c''  \|\sigma\| \, \lnorm{K}{-\ell}\, \lnorm{\phi}{-\ell}\,
 \lnorm{\phi'}{-\ell} (1+E)^{5\ell} d^{-3\ell-2}
 \sobnorm{f}{d}{3\ell+2}.
\end{multline}
This serves to define $\opd(f)$ on $\Sigma(E)$ for any $E$. Using
\cite[Lemma~2.6]{BDM:field_renorm}, we can extend this linear form to
$\sigspace{5 \ell + 1}$, and obtain another  constant $c'''$ such that
\begin{equation} \label{eqn:opdEstimate}
  \lsobnorm{\opd}{-5\ell-1}{d}{3\ell+2} \leq c'''
  \lnorm{K}{-\ell} \lnorm{\phi}{-\ell} \lnorm{\phi'}{-\ell}
  d^{-3\ell-2}.
\end{equation}
The extension is unique by density. It is also clear by construction that
$\opd(f)$ is linear in $f$ and in $\Pi$, i.e. multilinear in $\phi$, $\phi'$,
and $K$. Then, the estimate \eqref{eqn:opdEstimate} shows that $\opd$ is a
$\cinftyss$-valued distribution in the sense of Def.~\ref{def:cinftyssDist}.
\cmpqed\end{proof}

\subsection{Product expansions} \label{sec:ope}

We now prove the operator product expansion for a product of
pointlike fields, $\Pi\in\PhiProdLoc$, in the following form.

\begin{theorem} \label{thm:ope}
Let $\oprod \in \PhiProdLoc$, and let $\alpha\geq 0$. 
There exist $\ell>0$, $m \in\nbb$, $\gamma \geq 0$, and a
hermitean projector $p_\gamma:\cinftyss\to\Phi_\gamma$ onto $\Phi_\gamma$ such
that
\begin{equation*}
 \lsobnorm{	\opd - p_\gamma \opd }{-\ell}{d}{m} = o(d^\alpha)
  \quad \text{as }d \to 0.
\end{equation*}
\end{theorem}
This is a variant of \cite[Theorem~3.2]{Bos:product_expansions}. Note that the
approximation emerges into a more familiar form of operator product expansion
if $p_\gamma$ is written in a basis.

\begin{proof}
Again, we can assume $\Pi = K\otimes \phi \otimes \phi' \in \PhiProdL{\ell}$
for some $\ell>0$, where now $\phi,\phi' \in \PhiFH$. Further, after possibly
increasing $\ell$, we choose $k>0$ and approximating sequences $A_r$, $A_r'$ for
$\phi$, $\phi'$ as in Theorem~\ref{thm:fieldConverge}. 
Set $B_r = K\otimes A_r \otimes A_r'$. We define $m := 3 \ell+2$ and
$\gamma := (2k+\ell +3)(\alpha+m+1)$, 
and choose a hermitean projector $p_\gamma$ onto $\Phi_\gamma$. Now we estimate
for an as yet unspecified $\hat \ell$,
\begin{multline}
  \lsobnorm{\opd-p_\gamma \opd}{-\hat\ell}{d}{m}
\leq 
  \lsobnorm{\opd-\pd{B_r}}{-\hat\ell}{d}{m}
+  \lsobnorm{\pd{B_r}-p_\gamma \pd{B_r}}{-\hat\ell}{d}{m}
\\
+  \lnorm{p_\gamma}{-\hat\ell,\hat\ell}
\lsobnorm{\pd{(\Pi-B_r)}}{-\hat\ell}{d}{m}.
\end{multline}
Here $\lnorm{p_\gamma}{-\hat\ell,\hat\ell}$ is a constant independent of $r$ and
$d$, finite if $\hat\ell$ is large. We will show below that for large
$\hat\ell$,
\begin{align}
 \label{eqn:prodEst1}
\lsobnorm{\pd{(\Pi-B_r)}}{-\hat\ell}{d}{m} &= O(r
d^{-m}), \\ 
\label{eqn:prodEst2}
\lsobnorm{\pd{B_r}-p_\gamma \pd{B_r}}{-\hat\ell}{d}{m} &=
O(r^{-2k-\ell -2 } d^{-\ell-2} (r+d)^\gamma).
\end{align}
Setting $r(d)=d^{\alpha+m+1}$, and using $\gamma =
(2k+\ell +3)(\alpha+m+1)$, both terms above are of order $O(d^{\alpha+1})$,
of which the theorem follows.

To show Eq.~\eqref{eqn:prodEst1}, we write 
\begin{equation} \label{eqn:diffSplit}
\Pi-B_r = K\otimes (\phi-A_r) \otimes \phi' + K \otimes A_r \otimes
(\phi'-A_r') .
\end{equation}
For the first summand, we estimate by Proposition~\ref{pro:productExist}:
\begin{equation}
 \lsobnorm{ \pd{K\otimes (\phi-A_r) \otimes
 \phi'}}{-5\ell-1} {d}{m}
 \leq O(d^{-m}) \, \lnorm{K}{-\ell} \lnorm{\phi-A_r}{-\ell} \,
 \lnorm{\phi'}{-\ell} = O(r d^{-m}),
\end{equation}
as proposed. The second summand of Eq.~\eqref{eqn:diffSplit} has a similar
estimate, which combined gives Eq.~\eqref{eqn:prodEst1}.

For Eq.~\eqref{eqn:prodEst2}, we use the short-distance approximation of
Theorem~\ref{thm:projectorApprox} on the operator 
$A_r^P(s') := A_r(s'/2) A_r'(-s'/2) \in \afk(\ocal_{r+d})$, 
where $|s'|\leq d$,
and on its derivatives in $s'$. Using the estimates on the derivatives of
$A_r(t)$ and $A_r'(t)$ provided by Theorem~\ref{thm:fieldConverge}, this entails
that for large $\hat\ell$,
\begin{equation} \label{eqn:prodDerivEstimate}
   \lnorm{ \frac{d^n}{(ds')^n}\Big( A_r^P (s') - p_\gamma A_r^P (s')
   \Big)}{-\hat\ell} = O((r+d)^\gamma) O(r^{-2k-n}).
\end{equation}
Now we compute $\pd{B_r}-p_\gamma \pd{B_r}$, first on a fixed test function
$f\in\dcal(-d,d)$, $d\leq 1$, and on a fixed functional $\sigma \in \cinftys$. 
By Prop.~\ref{pro:productExist}, we have
\begin{multline}
  \sigma(\pd{B_r}(f)-p_\gamma \pd{B_r}(f))  = \int ds' h(s') K(s'-\imath 0),
  \\
  \text{where } h(s') = f(s') g(s'), \quad g(s')= \sigma\big(A_r^P(s') -
  p_\gamma A_r^P (s') \big).
\end{multline}
(Note that here $h$ is smooth, the only divergent factor is $K$. Therefore,
also, sharp energy-bounds of $\sigma$ do not play a role.) Using
Proposition~\ref{pro:analyticBoundary}, it follows that
\begin{equation} \label{eqn:pdEstimateStart}
  |\sigma(\pd{B_r}(f)-p_\gamma \pd{B_r}(f))|  \leq c \lnorm{K}{-\ell}
  d^{-\ell-2} \sobnorm{h}{d}{\ell+2}
\end{equation}
with a numerical constant $c$. For the Sobolev norm of $h$,
we can derive the following estimate by the Leibniz formula.
\begin{equation}
  \sobnorm{h}{d}{\ell+2} = \sobnorm{f g}{d}{\ell+2} \leq   2^{\ell+2}
  \sobnorm{f}{d}{\ell+2} \, \max_{0\leq n \leq \ell+2} d^n \sup_{t \in [-d,d]} |g^{(n)}(t)|.
\end{equation}
The derivatives of $g$ can be estimated by Eq.~\eqref{eqn:prodDerivEstimate}.
For $t \in [-d,d]$ one has
\begin{equation} \label{eqn:gDerivEst}
  |g^{(n)}(t)| \leq \lnorm{\sigma}{-\hat\ell} O((r+d)^\gamma) O(r^{-2k-n}),
\end{equation}
where the $O(\ldots)$ estimates are uniform in $\sigma$. Combining
Eqs.~\eqref{eqn:pdEstimateStart}--\eqref{eqn:gDerivEst}, we obtain
\begin{equation}
   \lsobnorm{\pd{B_r}-p_\gamma \pd{B_r}}{-\hat\ell}{d}{\ell+2} \leq
   O(d^{-\ell-2}) O((r+d)^\gamma) O(r^{-2k-\ell-2}),
\end{equation}
which gives Eq.~\eqref{eqn:prodEst2}.
\cmpqed\end{proof}

The bounds established are certainly not strict, in particular regarding the value of $\gamma$
(i.e., the number of approximation terms needed in the OPE). They might be
improved at the price of extra computational effort, but
this is not relevant for our purposes. Note however that the kernels $K$
introduce an extra divergence that might make more OPE terms necessary than in
the ``ordinary'' OPE version with $K=1$.

\subsection{Convolutions} \label{sec:convol}

In order to establish the existence of quantum inequalities, we need to analyse
distributions evaluated on certain convolutions of test functions, similar to
Eqs.~\eqref{eqn:fIntegral}--\eqref{eqn:remIntegral} in the free field case. Let us define
them, and establish their well-definedness. We remind the reader of the 
abbreviation $g_1 \syp g_2(s,s') = g_1(s+s'/2) g_2(s-s'/2)$, and of the notion
of skew-hermitean $\cinftyss$-valued distributions, which fulfill
$T(s')\st=T(-s')$.

\begin{lemma} \label{lem:convolWelldef}
Let $T$ be a $\cinftyss$-valued distribution. Then the bilinear map
\begin{align*}
  &\convolmap{T}: \dcal(\rbb) \times \dcal(\rbb)  \to
  \dcal(\rbb,\cinftyss), 
  \\
  & \convol{g_1,g_2}{T}(s) = \int ds' \, g_1
  \syp g_2 (s,s') T(s')
\end{align*}
 is well-defined; indeed, if $g_1,g_2 \in \dcal(-d,d)$,
 then $\supp \convol{g_1,g_2}{T} \subset (-d,d)$. Further,
\begin{align*}
 &\convolintmap{T}: \dcal(\rbb) \times \dcal(\rbb) \to \cinftyss,
 \\
  &\convolint{g_1,g_2}{T}= 
 \int ds \, U(s)  \big(  \convol{g_1,g_2}{T}(s)  \big) U(s)\st 
\end{align*}
is well-defined as a weak integral. Both $\convolmap{T}$ and
$\convolintmap{T}$ are linear in $T$. If $T$ is skew-hermitean, then
$\convolint{\bar g,g}{T}$ is hermitean for arbitrary $g \in \dcal(\rbb)$. For
any $m \in\nbb$ and $d>0$, one has the estimate
 \begin{equation*}
    \lsobnorm{\convolintmap{T}}{-\ell}{d}{m}
    \leq 2^{m+1}  \lsobnorm{T}{-\ell}{2d}{m}.
 \end{equation*}
\end{lemma}
The Sobolev norms of the bilinear maps are
understood here with respect to a product of
identical Sobolev norms on the two arguments.

\begin{proof}
First, $\convol{g_1,g_2}{T}(s)$ is well-defined since $g_1 \syp
g_2(s,\cdotarg)$ lies in $\dcal(-2d,2d)$ for each fixed $s$; and it is (weakly)
smooth in $s$ since $s \mapsto g_1 \syp g_2 (s,\cdotarg)$ is smooth in the $\dcal(\rbb)$
topology. The support properties are clear. Further, one sees that
\begin{equation} \label{eqn:convolNorm}
   \lnorm{\convol{g_1,g_2}{T}(s)}{-\ell}
   \leq \lsobnorm{T}{-\ell}{2d}{m} \, \sobnorm{g_1 \syp
   g_2(s,\cdotarg)}{2d}{m},
\end{equation}
which is locally bounded in $s$. Therefore, for each
$\sigma\in\cinftys$, the map 
\begin{equation}
 \rbb \to \cbb, \quad  s \mapsto \sigma\big( U(s)\,\convol{g_1,g_2}{T}(s) \,
 U(s)\st \big)
\end{equation} 
is continuous. Hence $\convolintmap{T}$ is well-defined as a weak integral.
Using the Leibniz rule and a change of variables, one finds 
\begin{align} \notag
\int ds\,\sobnorm{(g_1\syp g_2)(s,\cdot)}{2d}{m} &\le
\sum_{n=0}^m  \sum_{r=0}^n
\begin{pmatrix}n\\ r\end{pmatrix} d^r\|g_1^{(r)}\|_1 d^{n-r}\|g_2^{(n-r)}\|_1 \\
&\le  (2^{m+1}-1) \sobnorm{g_1}{d}{m} \sobnorm{g_2}{d}{m}.
\end{align}
Together with Eq.~\eqref{eqn:convolNorm}, this yields the estimate
\begin{equation}
    \lsobnorm{\convolintmap{T}}{-\ell}{d}{m}
    \leq 2^{m+1} \lsobnorm{T}{-\ell}{2d}{m},
\end{equation}
as proposed. Also, it
is clear in matrix elements that both $\convol{g}{T}$ and $\convolint{g}{T}$
are linear in $T$. If $T$ is skew-hermitean, one uses the identity
$\overline{\bar g \syp g (s,s')}=\bar g \syp g(s,-s')$ to conclude
$\convol{\bar g,g}{T}(s)\st=\convol{\bar g,g}{T}(s)$ and, in consequence,
$\convolint{\bar g,g}{T}\st = \convolint{\bar g,g}{T}$. 
\cmpqed\end{proof}

The estimates above show that our operator product expansion
for $\opd$, as established in Theorem~\ref{thm:ope}, can be transferred to
$\convolintmap{\opd}$. This is in fact the form of OPE we shall use for
establishing quantum inequalities.

\begin{corollary} \label{cor:opeConvol}
Let $\oprod \in \PhiProdLoc$, and let $\alpha\geq 0$. 
There exist $\ell>0$, $m \in\nbb$, $\gamma \geq 0$, and a
hermitean projector $p_\gamma:\cinftyss\to\Phi_\gamma$ onto $\Phi_\gamma$ such
that
\begin{equation*}
  \lsobnorm{
	\convolintmap{\opd - p_\gamma \opd} }{-\ell}{d}{m} = o(d^\alpha)
  \quad \text{ as } d \to 0.
\end{equation*} 
\end{corollary}

\subsection{Minimal approximating projectors} \label{sec:minapprox}

The operator product expansion allows us to approximate a given product $\Pi$
with a finite number of composite fields. It is important for our applications
to choose the minimal number of composite fields needed, so that
none of the approximation terms can be considered ``redundant". 

Let us introduce that notion of approximation by finitely many
terms more abstractly. This is similar, but not identical to
the analysis of normal products in \cite[Sec.~IV]{Bos:product_expansions}.

\begin{definition} \label{def:approxSpace}
Let $\Pi \in \PhiProdLoc$, and $\alpha \geq 0$. A
hermitean projector\footnote{Projectors in this space will always be assumed as
continuous.} $p$ in $\cinftyss$ with finite-dimensional image in $\PhiFH$
is called \emph{$\alpha$-approximating} for $\Pi$ if there are constants $\ell>0$ and $m \in \nbb$ 
such that
\begin{equation*} 
 \lsobnorm{  \convolintmap{ \opd-p \opd} }{-\ell}{d}{m} = o(d^\alpha)
 \quad \text{as } d \to 0.
\end{equation*}
\end{definition}

The operator product expansion in Corollary~\ref{cor:opeConvol} 
tells us that for any given $\alpha$, we can choose $\gamma$ large
enough such that any hermitean projector $p$ onto $\Phi_\gamma$ is
$\alpha$-approximating for $\Pi$. However, this is in a way an ``upper estimate'' to the OPE, since $\Phi_\gamma$ may
contain elements that are not actually needed for approximating the given
product. We will therefore minimize the approximating projector in a
well-defined sense.

This is done as follows. On the family of all $\alpha$-approximating projectors 
for a given product $\Pi$, we introduce a partial order by 
\begin{equation} \label{eqn:partialOrderDef}
   p_1 \leq p_2 \quad :\Leftrightarrow \quad
   (\img p_1 \subset \img p_2) \;\wedge\; (\ker p_1 \supset \ker p_2).
\end{equation}
Minimal elements with respect to this partially ordered set will be called
\emph{minimal $\alpha$-approxi\-mating projectors}. By dimensional arguments,
any decreasing sequence in the set must eventually become constant; so minimal
elements certainly exist, and can be constructed below each given $\alpha$-approximating
projector. However, there seems to be no reason why they should be unique. 

This is in contrast to the situation for normal
product spaces \cite[Sec.~IV]{Bos:product_expansions}, where the approximation
 property depends on $\img p$ only, i.e., any other projector onto the same
space would also be $\alpha$-approximating. In that case, one finds a unique minimal
approximating space of fields. In our situation, these stronger results do not
seem to follow, the main difficulty being that the convolution
$\convolintmap{\cdotarg}$ does not commute with projectors. 
This turns out not to be a problem however: 
Each minimal $\alpha$-approximating projector will give us a
nontrivial quantum inequality.

Let us summarize the main point of the above discussion:

\begin{proposition} \label{pro:miniExist}
 Let $\alpha \geq 0$ and $\Pi \in \PhiProdLoc$. 
 There exists at least one minimal $\alpha$-approximating projector $p$ for
 $\Pi$.
\end{proposition}

\section{Quantum inequalities} \label{sec:ineq}

We are now going to establish quantum inequalities as a consequence of the
operator product expansion above, and prove that they are
nontrivial as discussed in Sec.~\ref{sec:overview}.

\subsection{Existence of inequalities} \label{sec:ineq_existence}

In order to establish inequalities, we define a set of
products $\PhiPos \subset \PhiProdLoc$ which are
``classically positive'', namely a finite sum of absolute squares with positive-type
coefficients:
\begin{equation}
  \PhiPos := 
  \bigsetprop {\sum_j K_j \otimes \phi_j\st \otimes \phi_j }{
  K_j \in \kcal_+ ,\, \phi_j \in \PhiFH }.
\end{equation}
For any $\Pi \in \PhiPos$, the distribution $\opd$ is then skew-hermitean. [One
verifies this in matrix elements by the integral formula in
Prop.~\ref{pro:productExist}, using the relation
$\overline{K_j(z)}=K_j(-\bar z)$ for the positive-type kernels
$K_j$.]

Products from $\PhiPos$ now give rise to quantum inequalities. To
formulate these, we use the abbreviation $R := (1+H)^{-1}$.

\begin{theorem} \label{thm:qi}
Let $\oprod \in \PhiPos$ and $\alpha\geq 0$. 
Let $p$ be an $\alpha$-approximating projector for $\Pi$. There
exist $\ell > 0$, $m \in \nbb$, 
and a function $\epsilon:\rbb_+ \to \rbb_+$ of order $\epsilon(d)=o(d^\alpha)$
such that the following inequality between bounded
operators holds.
\begin{equation*}
  \forall d>0, \; g \in \dcal(-d,d): \quad 
   R^\ell \convolint{\bar g, g}{p\opd} R^\ell 
   \geq - \epsilon(d) ( \sobnorm{g}{d}{m})^2 \idop.
\end{equation*}
\end{theorem}

\begin{proof}
By Def.~\ref{def:approxSpace}, there exist $\ell$, $m$ and
$\epsilon(d)=o(d^\alpha)$ such that 
\begin{equation} \label{eqn:projApprox}
 \forall d>0, \,g \in \dcal(-d,d): \quad 
 \lnorm{ \convolint{\bar g,g}{p\opd-\opd} }{-\ell} \leq 
 \epsilon(d) \,( \sobnorm{g}{d}{m})^2 .
\end{equation}
Note here that, since $\opd$ is skew-hermitean, $\convolint{\bar g,g}{\opd}$ is
guaranteed to be hermitean by Lemma~\ref{lem:convolWelldef}. Since $p$ is
hermitean, the same is true for $\convolint{\bar g,g}{p\opd}$. The
expectation values of these expressions in positive functionals are
therefore real. Thus, for any $\rho \in \cup_E \Sigma(E)$, $\rho \geq 0$, we
obtain from Eq.~\eqref{eqn:projApprox},
\begin{equation} \label{eqn:rhoTriangle}
   \rho(\convolint{\bar g,g}{p \opd}) 
   - \rho(\convolint{\bar g,g}{\opd})
    \geq -\lnorm{\rho}{\ell} \epsilon(d) \,( \sobnorm{g}{d}{m})^2.
\end{equation}
Now $\Pi\in\PhiPos$ is of the form $\Pi = \sum_j K_j \otimes  \phi_j\st \otimes
\phi_j$. Due to energy-boundedness of $\rho$, we have by
Prop.~\ref{pro:productExist} and Lemma~\ref{lem:convolWelldef},
\begin{multline} \label{eqn:restIntegral}
   \rho(\convolint{\bar g, g}{\opd}) = \sum_j \int ds \int ds' \, \bar g \syp g
   (s,s') \, K_j(s'-\imath 0) \\
   \times \rho\Big(U(s)\, \phi_j\st(s'/2) \, \phi_j(-s'/2) \,
   U(s)\st\Big).
\end{multline}
Here
\begin{equation}
  G_j(s,s'-\imath 0) :=\rho\big(U(s) \phi_j\st(s'/2) \phi_j(-s'/2) U(s)\st \big)
  = \rho \big(\phi_j\st(s+s'/2) \phi_j(s-s'/2) \big)
\end{equation}
are positive-definite distributions, as they give
$\rho(\phi_j(g)\st\phi_j(g))$ when integrated with $\bar g \syp g$.
Also, a similar estimate as in Eq.~\eqref{eqn:fieldprodBounds} shows that $s
\mapsto G_j(s,\cdotarg)$ is uniformly bounded in $\lnorm{\cdotarg}{-\ell}$, and
also continuous since the energy-bounded state $\rho$ is analytic for
$U(s)$. Thus the products of $G_j$ with the positive-type
kernels $K_j$ are positive-definite as well (Prop.~\ref{pro:posMultiply}). 
Therefore, the expression in Eq.~\eqref{eqn:restIntegral} is non-negative. Setting $\hat \rho = \rho(R^{-\ell} \cdotarg R^{-\ell})$, we can thus reduce Eq.~\eqref{eqn:rhoTriangle} to 
\begin{equation}
  \hat \rho (R^\ell \, \convolint{\bar g, g}{p\opd} \, R^\ell) 
  \geq -  \,\epsilon(d) \,( \sobnorm{g}{d}{m})^2 \hat\rho(\idop).
\end{equation}
Here $R^\ell  \convolint{\bar g,g}{p\opd} R^\ell$ can be extended to a
bounded operator by Eq.~\eqref{eqn:projApprox}. Since
$\hat\rho$ can be chosen from a dense subset in the set of all positive functionals, the theorem now follows.
\cmpqed\end{proof}

The connection of the theorem with more usual forms of quantum inequalities
becomes clear when we write the projector $p$ in a basis:
\begin{equation} \label{eqn:pBasis}
p = \sum_{j=1}^n \sigma_j(\cdotarg) \phi_j, \quad
\text{where } \sigma_j \in \cinftys, \; \phi_j \in \PhiFH, \; \sigma_j(\phi_k) =
\delta_{j k}.
\end{equation}
Here we choose $\phi_j$ and $\sigma_j$ hermitean, which is possible since $p$
is hermitean. Then, the inequality in the theorem can be rewritten as
\begin{equation} \label{eqn:basisQi}
  \sum_{j=1}^n R^\ell \phi_j(f_j) R^\ell \geq - \epsilon(d)\,  (
  \sobnorm{g}{d}{m})^2 \,\idop,
\end{equation}
where the functions $f_1,\ldots,f_n$ are given by
\begin{equation} \label{eqn:fjDef}
  f_j (s) = \int ds'\,  \bar g \syp g (s,s')\, \sigma_j(\opd(s')) .
\end{equation}
These $f_j$ are actually of compact support,
namely $\supp f_j \subset (-d,d)$ if $g \in \dcal(-d,d)$, see
Lemma~\ref{lem:convolWelldef}. They are also
smooth, since $s \mapsto \bar g \syp g (s,\cdotarg)$ is differentiable in
the $\scal$-topology; so they are indeed proper test functions in
$\dcal(\rbb)$. Further, the $f_j$ are real-valued, which follows from
hermiticity of $\sigma_j$ and skew-hermiticity of $\opd$.

The inequality \eqref{eqn:basisQi} is of an asymptotic nature, inasmuch as only
the asymptotic behaviour of the remainder, $\epsilon(d)=o(d^\alpha)$, is known.
For the sake of concreteness, we may choose a fixed test function $g \in
\dcal(-1,1)$, and define a family of scaled functions $g_d(t)=d^{-1}g(t/d)$. For
these, $\sobnorm{g_d}{d}{m}$ is independent of $d$, so that the right-hand side of
Eq.~\eqref{eqn:basisQi} simplifies; the inequality is then valid as the
parameter $d$ of the family goes to $0$.

While the functions $f_j$ are real-valued, they are not
guaranteed to be pointwise positive, in contrast to the free field
situation \cite{FewEve:qei_flat}. That the positivity properties of
$f_j$ are a delicate issue is apparent since Eq.~\eqref{eqn:fjDef} has a strong
analogy to Weyl quantization. With $\tilde C_j$ being the Fourier transform of $C_j(s')=\sigma_j(\opd(s'))$,
one has
\begin{equation} \label{weylAnalogy}
f_j(s) = \int \frac{dp}{2\pi} \tilde C_j(p) \int ds'\; e^{\imath p s'} \;\bar g
\syp  g (s,s') = \int \frac{dp}{2\pi} \tilde C_j(p) W_g(s,p),
\end{equation}
where $W_g$ is the Wigner function associated with the ``state'' $g$,
\begin{equation}
W_g(s,p) = \int ds'\; e^{\imath p s'} \;\overline{ g (s+s'/2)}  g (s-s'/2). 
\end{equation}
Now the Wigner function \emph{cannot} be pointwise positive for compactly
supported $g$ \cite{Hud:Wigner}, so positivity of $f_j$ can only be expected in
special situations; see e.g., Prop.~\ref{pro:fPosDilation}.
 
Note that Eq.~\eqref{eqn:basisQi} is a far-reaching
generalization of the usual inequalities for squares of fields in free field
theory. In particular, the estimate will in general not be restricted to two
fields, such as the Wick square and the identity in Eq.~\eqref{eqn:usualQi}, but
will involve a possibly large number of fields smeared with different sampling
functions. One of the $\phi_j$ will typically be the identity operator, 
and another $\phi_j$ will typically be a
normal product in the sense of Zimmermann
\cite{Zim:Brandeis,Bos:product_expansions}. This term will usually be
distinguished as a highest-order field, relating e.g.~to scaling dimensions.
But there seems to be no guarantee that such highest-order field exists
uniquely, and even less that only two fields $\phi_1,\phi_2$ appear in the
inequality. Compared with the usual free-field situation, we also encounter a
remainder term $\epsilon(d)$ which seems unavoidable in this
context, but is of negligible order compared with the contributions of
the field operators, as we shall see below.

\subsection{Nontriviality} \label{sec:nontriv}

While Thm.~\ref{thm:qi} asserts that our construction yields a large
variety of valid quantum
inequalities, there remains the concern that they
could be trivial in the sense that the lower bound could also serve as
an upper bound, cf.~\cite{Few:categorical}. In particular,  
an inequality for a \emph{bounded} operator $A$ of the form $A \geq - \|A\|
\idop$ would be considered trivial. Since the exponent $\ell$ in
Eq.~\eqref{eqn:basisQi} is so large that all $R^\ell \phi_j R^\ell$ are
bounded, we might well encounter this situation: The left-hand side of
Eq.~\eqref{eqn:basisQi} might be dominated in norm by the remainder
$\epsilon(d)$. More generally, since Thm.~\ref{thm:qi} puts no further restrictions on the
projector $p$, it might also be possible that $p \opd$ contains
single ``redundant terms'' that are individually dominated by 
$\epsilon(d)$, and are thus essentially irrelevant.

We shall show now that if the approximating projector is chosen \emph{minimal},
these problems do not occur, and in this sense the inequality is nontrivial.

\begin{theorem} \label{thm:asymptNontriv}
  Let $\Pi \in \PhiPos$ and $\alpha \geq 0$. Let $p$ be a minimal
  $\alpha$-approximating projector for $\Pi$. 
  Let $V := \img p$.
  For sufficiently large $\ell>0$ and $m \in \nbb$,
  and for any hermitean projector $q:V \to V$, $q \neq 0$,
  it holds that 
  \begin{equation*}
    d^{-\alpha} \sup_{g \in \dcal(-d,d)}
    \frac{   \lnorm{\convolint{\bar g , g}{q p \opd}}{-\ell}}
    {(\sobnorm{g}{d}{m})^{2}} \not \to 0
    \quad
    \text{as $d\to 0$}.
  \end{equation*}
\end{theorem}

\begin{proof}
Suppose that $m$, $\ell$ and a hermitean projector $q:V \to V$ are given such
that
\begin{equation}
    d^{-\alpha} \sup_{g \in \dcal(-d,d)} (\sobnorm{g}{d}{m})^{-2}\,
    \lnorm{\convolint{\bar g , g}{q p \opd}}{-\ell} \to 0
    \quad
    \text{as $d\to 0$}.
\end{equation}
We will show $q=0$. First, we can use the polarization identity for the
quadratic form $\convolint{ g_1 , g_2}{q p \opd}$ in order to
show
\begin{equation}
    d^{-\alpha} 
    \lsobnorm{\convolintmap{q p \opd}}{-\ell}{d}{m} \to 0.
\end{equation}
The triangle inequality then yields
\begin{equation} \label{eqn:alphaTriangle}
    d^{-\alpha} \lsobnorm{\convolintmap{(1-q)p\opd-\opd}}{-\ell}{d}{m} 
    \\
    \leq
    d^{-\alpha} 
    \lsobnorm{\convolintmap{p\opd-\opd}}{-\ell}{d}{m} 
    +
    d^{-\alpha} 
    \lsobnorm{\convolintmap{q p \opd}}{-\ell}{d}{m} \to 0,
\end{equation}
since $p$ is $\alpha$-approximating; we suppose here that $m,\ell$ are
sufficiently large. Now \eqref{eqn:alphaTriangle} shows that $(1-q)p$ is also
$\alpha$-approximating for $\Pi$. It is clear that $(1-q)p \leq p$. Since
however $p$ is minimal, this implies $(1-q)p = p$. Thus $q=0$.
\cmpqed\end{proof}

Again, let us illustrate the content of the theorem by passing to a basis
representation of $p$, as in Eq.~\eqref{eqn:pBasis}. For the case
$q=\idop_V$, the theorem precisely shows that the left-hand side of
Eq.~\eqref{eqn:basisQi} does \emph{not} vanish in norm as fast as
$\epsilon(d)=o(d^\alpha)$. Further, choose $q$ specifically as $q = \sigma_k
\phi_k$ with fixed $k$. Then one obtains $\convolint{\bar g , g}{q p \opd} =
\phi_k(f_k)$, with $f_k$ as in Eq.~\eqref{eqn:fjDef}. 
Thus, Thm.~\ref{thm:asymptNontriv} provides us with a null sequence $(d_i)_{i
\in \nbb}$, a constant $c>0$, and a sequence of functions $g^{(i)} \in
\dcal(-d_i,d_i)$ with $\sobnorm{ g^{(i)} } {d_i}{m}=1$ such that
\begin{equation}
     \|R^\ell \, \phi_k(f_k^{(i)}) \, R^\ell\| \geq c \, (d_i)^\alpha 
     \quad \text{for all }i \in \nbb,
\end{equation}
where 
\begin{equation}
  f_k^{(i)} (s) = \int ds'\, \bar g^{(i)} \syp g^{(i)}(s,s') \,
  \sigma_k(\opd(s')) .
\end{equation}
So the field $\phi_k$ in the inequality~\eqref{eqn:basisQi} gives a
contribution that is large compared to the remainder $\epsilon(d)$.
Theorem~\ref{thm:asymptNontriv}, in full generality, shows that this conclusion
is true independent of the choice of basis.

We have argued in Prop.~\ref{pro:miniExist} that minimal $\alpha$-approximating
projectors $p$ exist for any product $\Pi$, and any $\alpha \geq 0$. So we
always obtain nontrivial quantum inequalities in the sense above. One might suspect here that
the minimization of the approximating projector $p$ could lead to
$p=0$, which might again be seen as trivial. While this is not the case
even in a simple free field example, we shall give a general argument that
shows that $p=0$ cannot occur, under a mild extra assumption.

\begin{theorem} \label{thm:pnotzero}
Let $\alpha \geq 0$, and $\Pi\in\PhiPos \backslash \{0\}$. Suppose
that the vacuum vector $\Omega$ is separating for the smeared fields $\phi(f)$, with
$\phi\in\PhiFH$ and $f\in\dcal(\rbb)$.
If $p$ is an $\alpha$-approximating projector for $\Pi$, then
$p \neq 0$.
\end{theorem}

We note that the condition of a separating vacuum vector is indeed a rather
weak one. It would suffice, for example, that
there exists a wedge region $\wcal$ such that $\Omega$ is cyclic for
$\afk(\wcal)$.

\begin{proof}
Suppose that $\alpha$ and $\Pi$ are given such that $p=0$ is
$\alpha$-approximating for $\Pi$. We will show $\Pi=0$. To that end, we
choose $\ell$ and $m$ sufficiently large, and pick a fixed positive test
function $g \in \dcal(-1,1)$. Then $g_d := d^{-1}g(d^{-1} \cdotarg)$
lies in $\dcal(-d,d)$, and $\sobnorm{g_d}{d}{m}=\sobnorm{g}{1}{m}$. 
Employing Def.~\ref{def:approxSpace}, we obtain
\begin{equation}
  \lnorm{ \convolint{\bar g_d, g_d}{ \opd }}{-\ell} \to 0 
  \quad \text{as } d \to  0.
\end{equation}
Evaluating the convolution integral in the vacuum state $\omega$ yields due to
translation invariance,
\begin{equation} \label{eqn:gdConfigspace}
  \int ds \, ds' \, \bar g_d \syp g_d (s,s') \; \omega(\opd(s'))
  \to 0 \quad \text{as } d \to 0.
\end{equation}
As argued in the proof of Thm.~\ref{thm:qi}, the
distribution $\omega(\opd(s'))$ is of positive type. Hence it is the Fourier
transform of a polynomially bounded positive measure $\mu$. 
With this information, we can rewrite Eq.~\eqref{eqn:gdConfigspace} as
\begin{equation} \label{eqn:gdMomentumspace}
  \int d\mu(p) \, | \tilde g_d (p) |^2   \to 0 \quad \text{as } d \to 0.
\end{equation}
However, as $d \to 0$, we have $|\tilde g_d (p)|^2 \to |\tilde g(0)|^2 >
0$ locally uniformly. Since $\mu$ is positive, we can conclude here that $\mu$
is the zero measure. So $\omega(\opd(s'))=0$ as a distribution. Using $\Pi\in\PhiPos$, we
have a representation
\begin{equation}
  0 = \omega(\opd(s')) = \sum_{j=1}^n K_j(s'-\imath 0) \, \omega(\phi_j\st
  U(-s'+\imath 0) \phi_j) \quad \text{with } K_j \in \kcal_+, \; \phi_j \in
  \PhiFH.
\end{equation}
Since all summands are of positive type, each of them must vanish individually;
and clearly, also their analytic continuations must vanish. Thus, for any $j$,
we have either $K_j=0$ or $\omega(\phi_j\st
U(-s') \phi_j)=0$. But the latter implies $\|\phi_j(f) \Omega\|=0$ for any $f$
of compact support; thus $\phi_j(f)=0$ by assumption, and ultimately
$\phi_j=0$ by passing to a delta sequence. In total, this means $\Pi=0$.
\cmpqed\end{proof}

One might also be concerned that $p$ might project only onto multiples of the
identity. Again, this does not occur in the simple example of the Wick square
of the free field, as discussed in Sec.~\ref{sec:overview}. In general, we
conjecture, but have not proved, that in this case all fields appearing in the
product $\Pi$ must be multiples of the identity. At the very least, one can show
that the projector may be taken to be of the form $p=\omega(\cdotarg) \idop$, where
$\omega$ is the vacuum state. If this $p$ is indeed $\alpha$-approximating for
$\Pi$, the normal product of $\Pi$ can be defined by point splitting, and
vanishes identically. So this does not seem to be a case of great interest.

\subsection{Mesoscopic bounds}\label{sec:mesoscopic}

The inequalities derived above involve a remainder term that vanishes in
the small distance limit. Here, we discuss how the remainder can be reduced for test
functions of fixed supports, essentially by forming a Riemann integral of the
bounds at short distance.

Let $\chi\in\dcal(-1,1)$ and $f\in\dcal(-d,d)$ be fixed nonnegative
functions. We set
$\chi_\lambda(s)=\lambda^{-1}\chi(s/\lambda)$ for $\lambda\in(0,1]$. As
in Thm.~\ref{thm:qi}, we suppose $p$ to be an $\alpha$-approximating projector for $\oprod \in \PhiPos$,
with $\alpha\geq 0$. The basic inequality of
Thm.~\ref{thm:qi}, applied to $\chi_\lambda$, entails
\begin{equation}
   R^\ell \convolint{\chi_\lambda , \chi_\lambda }{p\opd} R^\ell 
   \geq - \epsilon(\lambda) ( \sobnorm{\chi_\lambda}{\lambda }{m})^2
\idop  =- \epsilon(\lambda) ( \sobnorm{\chi}{1 }{m})^2 \idop
\end{equation}
for suitable $\ell>0$ and $m\in\nbb$, where $\epsilon(\lambda)=o(\lambda^\alpha)$. 
Applying a time-translation through $\lambda k$, multiplying by
$\lambda f(\lambda k)$ and summing, we find
\begin{align}\notag
   \sum_{k\in\zbb} \lambda f(\lambda k)
 U(\lambda k)\convolint{\chi_\lambda , \chi_\lambda }{p\opd} U(\lambda k)^*  
  & \geq - \epsilon(\lambda) ( \sobnorm{\chi}{1 }{m})^2 \sum_{k\in\zbb}
\lambda f(\lambda
k) R^{-2\ell}\\ 
&\geq - \epsilon(\lambda) ( \sobnorm{\chi}{1 }{m})^2 (\|f\|_1 +\lambda
\|f'\|_1) R^{-2\ell}
\end{align}
Passing to a basis representation, we may
rewrite this inequality in the form
\begin{equation}\label{eq:mesoqi}
\sum_{j=1}^n  \phi_j(F_{j,\lambda}) \geq - 2\epsilon(\lambda) ( \sobnorm{\chi}{1 }{m})^2
\sobnorm{f}{d}{1} R^{-2\ell}
\end{equation} 
for $\lambda\le d$ where 
\begin{equation}
F_{j,\lambda}(s) =   \sum_{k\in\zbb} \lambda f(\lambda k) \int ds'
\sigma_j(\opd(s')) \chi_\lambda\diamond\chi_\lambda (s-\lambda k,s') .
\end{equation}
Owing to the support properties of $\chi_\lambda$, at most
two terms contribute to the sum on $k$ for each fixed $s$; moreover, $F_{j,\lambda}\in\dcal(-d,d)$.

In any fixed state in $\cinftys$
the expectation value of the right-hand side of \eqref{eq:mesoqi}
can be made arbitrarily small by reducing
$\lambda$, while the behaviour of the terms on the right-hand side is
determined by the asymptotic behaviour of the $F_{j,\lambda}$, regarded
as compactly supported distributions. In the unlikely event that each
$F_{j,\lambda}$ converged to a limit in the weak-$*$ topology on $\ecal'(\rbb)$, we would have
established a quantum inequality without remainder term. 
It may be useful to give two examples. If the OPE coefficient
$\sigma_j(\opd(s'))$ is smooth, then convergence does occur, with
\begin{equation}
F_{j,\lambda}\to \sigma_j(\opd(0)) (\|\chi\|_1)^2 f \qquad\text{in
$\ecal'(\rbb)$ as $\lambda\to 0$.}
\end{equation}
(To see this, one integrates against $u(s)$ and observes that the $k$'th
summand is subject to only an $O(\lambda^2)$ error if $u(s)\sigma_j(\opd(s'))$
is replaced by $u(\lambda k)\sigma_j(\opd(0))$; as there at most
$O(\lambda^{-1})$ nonzero summands the result follows by a simple calculation.)
On the
other hand, if $\sigma_j(\opd(s'))=(\imath\pi)^{-1}/(s'-\imath 0)$, we find 
\begin{equation}
\lambda F_{j,\lambda} \to (\|\chi\|_2)^2 f\qquad\text{in
$\ecal'(\rbb)$ as $\lambda\to 0$.}
\end{equation}
(Note that it is the $L^2$-norm that appears here, in contrast to the
first example.)

In general, therefore, it cannot be expected that all of the $F_{j,\lambda}$ 
converge as $\lambda\to 0$. Nonetheless, as in the second example, its
leading order behaviour in $\lambda$ can be identified as follows. 
\begin{proposition} \label{pro:Fjl_leading}
Let $q$ be the order of the germ of $\sigma_j(\opd(s'))$ at $s'=0$
and define
\begin{equation}
\eta_j(\lambda) = \int ds'\, \sigma_j(\opd(s')) (\chi_\lambda * \hat{\chi}_\lambda)(s'),
\end{equation}
where $\hat{\chi}(s')=\chi(-s')$. If
$\lambda^{-q}\eta_j(\lambda)^{-1}=o(1)$ as $\lambda\to 0$ then
\begin{equation}\label{eq:F_cvgnce}
F_{j,\lambda}/\eta_j(\lambda) \to f \qquad\text{in
$\ecal'(\rbb)$ as $\lambda\to 0$}.
\end{equation}
In particular, this is satisfied if $\sigma_j(\opd(s'))$ has a
scaling limit of degree $\beta<0$ and $q=\lceil-1-\beta\rceil$. 
\end{proposition}
Here, the {\em order of the germ} of $\sigma(\opd(s'))$ at $s'=0$ is the
minimal $q\in\nbb_0$ for which there are $\lambda_0>0$ and $C>0$ such that
$|\int ds'\,\sigma(\opd(s'))u(s')|\le C\sum_{r=0}^q \sup |u^{(r)}|$ for
all $u\in\dcal(-\lambda_0,\lambda_0)$. The notion of {\em scaling limit} is taken from
\cite{FreHaa:covariant_scaling}: namely, the scaling limit exists if 
there exists a monotone positive function $N(\lambda)$ for which
\begin{equation}\label{eq:scaling_limit_def}
N(\lambda)\int ds'\,\sigma_j(\opd(s')) u_\lambda(s') \to S(u)
\end{equation}
for all $u\in\dcal(\rbb)$, with a nonzero limit for at least one $u$.
Under these circumstances, $S$ is a homogeneous distribution,
i.e., $S(u_\lambda)=\lambda^\beta S(u)$, with degree $\beta\in\rbb$
determined by
\begin{equation}\label{eq:beta_def}
\lim_{\lambda'\to 0} \frac{N(\lambda')}{N(\lambda\lambda')} = \lambda^{\beta}.
\end{equation}
(Our definition of the degree coincides with that of \cite[Ch.~I
Sec.~1.6.]{GelShi:genfun1} and differs from  
\cite{FreHaa:covariant_scaling}.) If $\beta<0$, for example, the
distribution $(s'-i0)^{\beta}(\log s'-i0)^\gamma$ has a
scaling limit of degree $\beta$ and (germ) order $\lceil-1-\beta\rceil$, and
therefore meets the criteria stated. 
  
\begin{proofwithremark}{of Prop.~\ref{pro:Fjl_leading}}
We choose
$\lambda_0\in(0,1]$ sufficiently small that $\sigma(\opd(s'))$ has order $q$ on
$(-2\lambda_0,2\lambda_0)$, and assume henceforth that
$0<\lambda<\lambda_0$. 
As in the second example above, we integrate $F_{j,\lambda}$ 
against $u\in\ecal(\rbb)$ and approximate $u(s)$ by $u(\lambda k)$
in the $k$'th summand, to obtain
\begin{align}\notag
\int ds\, F_{j,\lambda}(s)u(s) &= 
\int ds'\, \sigma_j(\opd(s')) \sum_{k\in\zbb} \lambda f(\lambda k)\int
ds\, \chi_\lambda\diamond\chi_\lambda(s,s')u(s+\lambda k) \\ 
&= \eta_j(\lambda) \sum_{k\in\zbb} \lambda f(\lambda k) u(\lambda k)  +
R_{j,\lambda},
\end{align}
where
\begin{equation}
R_{j,\lambda} = \int ds'\, \sigma_j(\opd(s')) 
\sum_{k\in\zbb} \lambda f(\lambda k)\int
ds\, \chi_\lambda\diamond\chi_\lambda(s,s') [u(s+\lambda k)-u(\lambda
k)] .
\end{equation}
Now $R_{j,\lambda}$ is, at worst, of order $O(\lambda^{-q})$ as $\lambda\to 0$, as is easily
seen using the estimate
\begin{equation}
\sup_s \big|\int ds'\, \sigma_j(\opd(s')) \chi_\lambda\diamond\chi_\lambda(s,s')\big|
\le \frac{C}{\lambda^{q+2}};
\end{equation}
and the facts that 
(i) the sum contains at most $O(\lambda^{-1})$ nonzero terms; (ii) the $s$-integral
extends over the region $[-\lambda,\lambda]$. This establishes
\begin{equation}\label{eq:eta_leading}
\int ds\,F_{j,\lambda}(s)u(s) =\eta_j(\lambda)\left(\int ds\,f(s)u(s) +
O(\lambda) \right) + O(\lambda^{-q})
\end{equation}
as $\lambda\to 0$, from which \eqref{eq:F_cvgnce}
follows immediately.

Now suppose that $\sigma_j(\opd(s'))$ has a
scaling limit of degree $\beta<0$. 
It is easy to see that \eqref{eq:scaling_limit_def} implies $N(\lambda)\eta_j(\lambda)\to
S(\chi*\hat{\chi})$. The spectrum condition
entails that $S= C(\imath(\cdot-\imath 0))^{\beta}$, where the nonzero constant
$C$ is real owing to hermiticity (cf.\ the proof of 
Prop.~\ref{pro:oneDTerm} below). As $\beta<0$, we may verify directly that
$S(\chi*\hat{\chi})\not=0$, that $N(\lambda)$ is necessarily monotone decreasing and
vanishing as $\lambda\to 0$. Thus $\eta_j(\lambda)\to\pm \infty$ depending
on the sign of $C$. Moreover, Eqs.~\eqref{eq:scaling_limit_def}
and~\eqref{eq:beta_def} entail
\begin{equation}
\lim_{\lambda'\to
0}\frac{(\lambda')^q\eta_j(\lambda')}{(\lambda\lambda')^q\eta_j(\lambda\lambda')}
= \lambda^{-\beta-q}.
\end{equation}
By hypothesis, $\sigma_j(\opd(s'))$ has order $q=\lceil-1-\beta\rceil$
(as does $S$). Thus 
$-\beta-q>0$ and we deduce that $\lambda^{-q}\eta_j(\lambda)^{-1}\to 0$ as
$\lambda\to 0$. 
\cmpqed\end{proofwithremark}

The significance of this result becomes clear in the situation where one of the composite fields, say
$\phi_1$, is identified as a field of particular interest, e.g., the
normal product. By hermiticity of the projection $p$, $\eta_1$ is
real-valued; the hypothesis of Prop.~\ref{pro:Fjl_leading} requires that
$|\eta_1(\lambda)|\to\infty$ as $\lambda\to 0$. If, in fact, $\eta_1(\lambda)\to+\infty$,
we may divide the quantum inequality \eqref{eq:mesoqi} by this
factor to obtain a bound
\begin{equation}
\phi_1(F_{1,\lambda}/\eta_1(\lambda)) + \frac{1}{\eta_1(\lambda)}\sum_{j=2}^n \phi_j(F_{j,\lambda}) \geq -
\frac{2\epsilon(\lambda)}{\eta_1(\lambda)} ( \sobnorm{\chi}{1 }{m})^2
\sobnorm{f}{d}{1} R^{-2\ell}
\end{equation}
for $\lambda<d$. (If $\eta_1(\lambda)\to-\infty$ we simply reverse the sign of $\phi_1$
and hence $\sigma_1(\opd(s'))$ and 
$\eta_1$ to obtain the same result; the possibility that $\eta_1$
oscillates in sign as $\lambda\to 0$ can be excluded if the scaling limit exists.)
In this form, it is clear that the remainder term may be diminished by reducing
$\lambda$, at the possible cost of increasing the magnitude of the terms in
composite fields with $j\ge 2$ (if $\eta_j(\lambda)$ grows more rapidly
than $\eta_1(\lambda)$).
Moreover, the expectation value of the
first term tends to that of $\phi_1(f)$ as $\lambda\to 0$ for any state
in $\cinftys$.

Further progress is only possible with more detailed information
regarding the (germs of the) OPE coefficient distributions. Nonetheless, we expect that
the results presented here will be of use in the context of particular models.

\section{Scaling limits and dilation covariance} \label{sec:scaling}

For a concrete interpretation of our quantum inequalities, it is of
particular interest to investigate the detail structure of the sampling
functions with which the composite fields are smeared, e.g.~the functions $f_j$
in Eq.~\eqref{eqn:fjDef}. For example, one is interested whether they are
pointwise positive, or at least ``mostly positive'' in a
well-defined sense. Of course, these properties depend crucially on the
structure of the OPE coefficients involved, about which little is known in the
general case. The most reasonable approach therefore seems to investigate those
properties under more restrictions on the theory.

In the preceding section, our approach was to approximate a given sampling
function with a Riemann sum; this relied on some assumptions on the
behavior of the OPE coefficients in the small, and was tied to a choice of
basis in the field spaces. In the following, we want to take a different
approach: We investigate the structure of sampling functions in a
restricted class of quantum field theories, 
namely in the presence of dilation symmetries. 
While for a realistic description of microphysics, one would not consider dilation covariant quantum field theories, this case is still important as an idealization at short scales. Namely, in the short-distance regime, quantum field theories should be approximated by a
\emph{scaling limit theory}, which indeed possesses a dilation symmetry.

Let us briefly sketch how the scaling limit of quantum field theories fits into
our context. It has been shown by Buchholz and Verch
\cite{BucVer:scaling_algebras} that scaling limits can be formulated very
naturally on the level of local algebras. Every quantum field theory possesses
a scaling limit in this sense, although it might not be unique. The limit
theory is, under a suitable choice of limit states, covariant under a strongly
continuous unitary representation of the dilation group
\cite{BDM:field_renorm}. However, the structure of these dilation unitaries may
be very intricate, acting on a nonseparable Hilbert space. (See also
\cite{BDM:dilations}.)

In \cite{BDM:field_renorm}, it was shown that this picture is compatible with
the usual notion of field renormalization: If the original algebraic theory fulfils
a slightly sharpened version of Def.~\ref{def:microPhase}, then the limit
theory fulfils Def.~\ref{def:microPhase} too; and pointlike fields in the
original theory converge, under a multiplicative renormalization scheme, to
pointlike fields in the limit theory. In a certain sense, the projectors
$p_\gamma$ onto $\Phi_\gamma$ converge to corresponding projectors
$p_\gamma^{(0)}$ in the limit theory. Also, this scheme is compatible with
products of pointlike fields and operator product expansions. Thus one can
expect that the structures exhibited in Sec.~\ref{sec:ineq} properly
converge in the scaling limit, and yield quantum inequalities in the limit
theory.

Our aim here is neither to describe this passage to the limit theory in detail,
nor to treat all possible cases of dilation group representations that may
appear in the limit. Rather, we take the above as a motivation to investigate
quantum inequalities in dilation covariant theories, and to show in certain
simple cases that stricter classification results on the form of quantum
inequalities can be achieved.

In the remainder of this section, we will therefore assume that our theory
$\afk$ has a dilation symmetry; i.e., that there exists a strongly continuous
unitary representation $\lambda \mapsto U(\lambda)$ of the dilation group on $\hcal$,
which is compatible with the Poincaré group representation, and acts on the
local algebras in the usual geometric way. The adjoint action of $U(\lambda)$
can then be extended to $\cinftyss$, where we write $\delta_\lambda \phi =
U(\lambda) \phi U(\lambda)\st$ in the weak sense. The spaces $\Phi_\gamma$ are
invariant under $\delta_\lambda$ \cite[Sec~IV]{Bos:short_distance_structure}. We shall
now consider the action of $\delta_\lambda$ on the structures considered so
far, and introduce some definitions for convenience.

\begin{definition}\label{def:dilCovar}
A quadratic form $\phi \in \cinftyss$ is called \emph{dilation
covariant} if, with some $\beta \in \rbb$,
\begin{equation*}
  \delta_\lambda \phi = \lambda^\beta \phi \quad \text{for all }\lambda>0.
\end{equation*}
A product $\Pi \in \PhiProd$ is called dilation
covariant if, with some $\beta \in \rbb$,
\begin{equation*}
  \delta_\lambda \opd(s) = \lambda^\beta \opd(\lambda s) 
   \quad \text{for all }\lambda>0, \text{in the sense of distributions.}
\end{equation*}
A projector $p$ in $\cinftyss$ is called dilation
covariant if
\begin{equation*}
  \delta_{1/\lambda} \circ p \circ \delta_\lambda \restrict \afk(\ocal_1) 
  = p \restrict \afk(\ocal_1) 
 \quad \text{for all } 0< \lambda \leq 1,
\end{equation*}
where $\ocal_1$ is the standard double cone of radius 1.
\end{definition}

Note that the restriction to
$\afk(\ocal_1)$ in the definition of dilation covariant projectors is
unavoidable if we want $p$ to be norm-bounded on $\boundedops$. Namely,
suppose that $\delta_{1/\lambda} \circ p \circ \delta_\lambda(A) = p(A)$ for
all $A \in \boundedops$ and $0 < \lambda \leq 1$, and hence for all $\lambda$
by the group relation. Since $\delta_\lambda$ acts as a norm
isomorphism on $\boundedops$, norm-boundedness of $p$ would lead to
$\delta_{\lambda}$ being uniformly bounded on the finite dimensional space $\img p$, 
both for $\lambda \to 0$ and for $\lambda \to \infty$, 
which would exclude that $\img p$ contains fields with nonzero scaling
dimension.

Dilation covariant products can easily be constructed, e.g. by choosing
dilation covariant fields $\phi_1$, $\phi_2$, and setting
$\Pi = (\imath z)^{-\beta'} \otimes \phi_1 \otimes \phi_2$ with some
$\beta'\geq 0$. If $\Pi$ and $p$ are both dilation covariant, Def.~\ref{def:dilCovar} implies
that
\begin{equation} \label{eqn:projProductCovar}
   \delta_{\lambda} \,p \,\opd(s') = \lambda^\beta \,p\, \opd(\lambda s')
   \quad
   \text{for } 0< \lambda \leq 1\,
   \text{ and for } s'\in[-1,1];
\end{equation}
that is, the equation holds when evaluated on test functions with support in
$[-1,1]$. This follows by approximating $\opd$ with sequences of bounded local
operators, as in the proof of Thm.~\ref{thm:ope}.

We will now consider the form of quantum inequalities in our case, that is,
investigate the structure of minimal approximating projectors $p$ and their
subprojectors. We shall restrict here to the simplest case, where one deals
with one-dimensional subrepresentations of $\delta_\lambda$. In this case, we
can find a full classification of our quantum inequality terms.

\begin{proposition} \label{pro:oneDTerm}
Let $\Pi \in \PhiPos$ be dilation covariant. Let $p$ be a
\emph{one-dimensional} dilation covariant projector in $\cinftyss$. Then, there
exist a dilation covariant field $\phi \in \PhiFH$ and $\beta \in \rbb$ such
that
\begin{equation*}
  p \opd (s') = (\imath(s'-\imath0))^{\beta} \phi
  \quad \text{on the interval } (-1,1).
\end{equation*}
\end{proposition}

\begin{proof}
We choose $\phi \in \PhiFH$ and $\sigma \in \cinftys$ such that
$p=\sigma(\cdotarg)\phi$. Since $\sigma(\phi)=1$, and since $\phi$ can be
approximated by bounded operators as in Thm.~\ref{thm:fieldConverge}, we can
find $A \in \afk(\ocal_1)$ such that $\sigma(A)=1$. Using that $p$ is dilation
covariant, we obtain
\begin{equation}
  \sigma(\delta_\lambda A) \,\delta_{1/\lambda} \phi = \sigma(A)\phi = \phi
  \quad \text{for all } 0 < \lambda \leq 1,
\end{equation}
and thus
\begin{equation}
  \delta_{\lambda} \phi = \sigma(\delta_\lambda A)\phi = c(\lambda) \phi
  \quad \text{for all } 0 < \lambda \leq 1.
\end{equation}
Here the $\cbb$-valued function $c(\lambda)$ is continuous in $\lambda$ and
fulfils $c(1)=1$, $c(\lambda)c(\lambda')=c(\lambda\lambda')$ if
$\lambda,\lambda' \in (0,1]$. This suffices to conclude that there exists a $\beta_1\in\cbb$ such
that
\begin{equation}
  c(\lambda) = \lambda^{\beta_1} \quad \text{for all } 0<\lambda\leq 1.
\end{equation}
Due to the group relation, we then obtain for \emph{all} $\lambda \in \rbb_+$,
\begin{equation}
   \delta_\lambda \phi = \lambda^{\beta_1} \phi.
\end{equation}
Splitting $\phi = \phi_R + \imath \phi_I$ into real and imaginary parts, we
note that $\delta_\lambda$ preserves this splitting, which means that $\beta_1$
must be real. So $\phi$ is dilation covariant. Inserting into
Eq.~\eqref{eqn:projProductCovar}, we arrive at
\begin{equation}
  \sigma(\opd(s')) = \lambda^{\beta_2-\beta_1} \sigma(\opd(\lambda s'))
  \quad 
  \text{in the sense of }\dcal(-1,1)',
\end{equation} 
where $\beta_2\in\rbb$ is the exponent relating to $\Pi$.
Using the right-hand side as a definition for $|s'|>1$, we can construct a
homogeneous distribution\footnote{%
As mentioned in Sec.~\ref{sec:mesoscopic}, alternative conditions that
force a distribution in the scaling limit to be homogeneous are
discussed in \cite{FreHaa:covariant_scaling}.
}
$D \in \dcal(\rbb)$ of degree
$\beta:=\beta_1-\beta_2$ such that
\begin{equation} \label{eqn:DEqualsSigma}
  D(s') = \sigma(\opd(s'))
  \quad 
  \text{in the sense of }\dcal(-1,1)'.
\end{equation}
The homogeneous distributions of one variable are however fully
classified (cf.~\cite[Ch.~I Sec.~3.11.]{GelShi:genfun1}): They are of the form
\begin{equation}
 D(s') = c_+ (s'+\imath 0)^\beta + c_- (s'-\imath 0)^\beta \quad \text{with }
 c_\pm\in\cbb. 
\end{equation} 
We can further restrict the possible form of $D$. Since $\sigma$ can be
approximated by energy-bounded functionals $\sigma_E$, and $\sigma_E(\opd(s'))$
has an analytic continuation to the lower half-plane, the only singular
direction (in the sense of wave front sets) of $\sigma(\opd(s'))$ at 0 can be
the positive half-line. Since the wave front set is determined locally,
Eq.~\eqref{eqn:DEqualsSigma} entails that $c_+=0$. Absorbing a factor $
\imath^{-\beta} c_-$ into the field $\phi$, we finally obtain
\begin{equation}
  p \, \opd(s') = \big(\imath (s'-\imath 0)\big)^\beta \phi \quad 
  \text{on the interval }(-1,1),
\end{equation}
as proposed.
\cmpqed\end{proof}

Now in the above situation, we can easily describe the quantum inequality terms
that arise. One finds for any $g \in \dcal(-1,1)$,
\begin{equation} \label{eqn:oneDTerms}
   \convolint{\bar g,g}{p\opd} = \phi(f) \quad
   \text{with}\;
   f (s) = \int ds' \,(\imath(s'-\imath 0) )^\beta\, \bar g \syp
   g(s,s') \;.
\end{equation}
This expression would not represent the entire quantum inequality, as
approximating projectors will typically not be one-dimensional. Rather,
\eqref{eqn:oneDTerms} would represent \emph{one} of the summands of the
inequality in Eq.~\eqref{eqn:basisQi}. In typical cases, one may expect that
there exists a distinguished highest-order term in the operator product
expansion, which corresponds to the ``normal product'' part of $\Pi$, and
which is described by a one-dimensional dilation covariant projector as above.

Note that Prop.~\ref{pro:oneDTerm} determines the field $\phi$ uniquely.
In particular, for $\beta \leq 0$, requiring the distributional
factor to be of positive type fixes the phase factor of $\phi$. 
While other conditions might be used to restrict this phase factor, 
such as demanding that $\phi$ be hermitean, 
the quantum inequalities give a stronger restriction that even fixes a $\pm$ sign in $\phi$. 
In this sense, our quantum inequalities can be used to 
distinguish the normal square of a field from its negative; squares of fields
retain certain aspects of positivity in the quantum case.

Let us further investigate the structure of the smearing function $f$ obtained
in Eq.~\eqref{eqn:oneDTerms}. We assume for a moment that $g$ is
real-valued, and thus $\bar g \syp g(s,s')$ is symmetric in $s'$. 
By a standard computation \cite[Ch.~I §3 Nr.~8]{GelShi:genfun1}, one obtains
the following simplified expressions in terms
of convergent integrals:
\begin{align} \label{eqn:fInt_BPositive}
 f(s) &= 2 \cos\frac{\beta\pi}{2} \int_{0}^\infty ds'\,
 (s')^{\beta} g\syp g(s,s') \quad &\text{for } \beta>-1,
\\ \notag
 f(s) &= 2  \cos\frac{\beta\pi}{2}  
 \int_{0}^\infty
 ds'\, (s')^{\beta} \Big( g\syp g(s,s')
 \\ \label{eqn:fInt_BNegative}
 &\quad  - \sum_{k=0}^{[(-\beta-1)/2]}  
 \frac{1}{(2k)!} \frac{\partial^{2k}g \syp
  g}{(\partial s')^{2k}}  \big|_{s'=0} s'^{2k} \Big) &\text{for } \beta<-1,\;
  |\beta| \not\in 2 \nbb+ 1,
\\ \label{eqn:fInt_BEvenInt}
 f(s) &=  \frac{(-1)^{k} \pi}{(2k)!} \; \frac{\partial^{2k}g \syp
  g}{(\partial s')^{2k}}  \big|_{s'=0} &\text{ for } \beta = -2k-1, \; k \in
  \nbb_0.
\end{align}

Using these explicit characterizations, we can directly investigate the
positivity properties of the function $f$. For reasons of simple
interpretation, it would be convenient if the $f(s)$ are positive at each
$s$. We can give some sufficient conditions to this end.
\begin{proposition} \label{pro:fPosDilation}
  Let $g \in \dcal(\rbb)$, $\beta \in \rbb$, and $f$ be
  given as in Eq.~\eqref{eqn:oneDTerms}. If any of the following conditions is
  fulfilled, it follows that $f(s) \geq 0$ for all $s \in \rbb$.
  \begin{enumerate}
    \localitemlabels
    \item \label{itm:betaPos} 
      $-1 < \beta \leq 1$, and $g(t) \geq 0$ for all $t \in
    \rbb$.
    \item \label{itm:betaZero}
      $\beta = -1$, and $g$ is real-valued.
    \item \label{itm:betaABitNeg} 
        $-3 < \beta < -1$, $\supp g$ is a connected
    interval $I$, and $g$ is logarithmically concave within~$I$.
   \end{enumerate}
\end{proposition}

\begin{proof} 
The case \ref{itm:betaPos} follows immediately from
Eq.~\eqref{eqn:fInt_BPositive}. In case~\ref{itm:betaZero}, we
obtain $f(s)=\pi g(s)^2$ from Eq.~\eqref{eqn:fInt_BEvenInt}, which
yields the result. For~\ref{itm:betaABitNeg}, observe that in this
case Eq.~\eqref{eqn:fInt_BNegative} reads
\begin{equation}
 f(s) = 2 \,\big| \cos\frac{\beta\pi}{2} \,\big|  
 \int_{0}^\infty
 ds'\, (s')^{\beta} \Big( g(s)^2  - g(s+s'/2)g(s-s'/2) \Big).
\end{equation}
Now the concavity of $t \mapsto \log g(t)$ precisely implies that $g(s)^2  \geq
g(s+s'/2)g(s-s'/2)$ for any $s$ and $s'$.
\cmpqed\end{proof}

The case $\beta=-1$ corresponds to the leading order of the OPE in the Wick
square of a massless free field theory, as discussed in Sec.~\ref{sec:overview}.
Our main interest is therefore in the case where $\beta$ is near $-1$,
which might be expected in asymptotically free theories. 
This realm is covered in the above proposition.
In models, it might be possible to exploit the choice of positive-type kernels
$K_j$ in the definition of $\Pi$ in order to arrive at precisely the case
$\beta=-1$, so that the function $f=g^2$ has a simple interpretation. We do
however not investigate this possibility in detail here.

In more generality, for any $\beta \leq 0$, we can at least state the following
more qualitative result: Since $(\imath (s'-\imath 0))^{\beta}$ is of positive
type, one finds 
\begin{equation}
   \int ds \, f(s) \geq 0,
\end{equation}
so $f$ has at least a non-negative average, regardless of the choice of
$g$. A bit more generally, one can deduce G{\aa}rding
inequalities for $f$, similar to those familiar from quantum mechanics \cite{EFV:qm_inequalities}: For suitable
test functions $\chi$, one has
\begin{equation}
  \int \chi(s) f(s) ds \geq -c_\chi (\|g\|_2)^2.
\end{equation}
Thus positivity of the test function $f$ is preserved at least in a generalized
sense.

\section{Conclusions and Outlook} \label{sec:outlook}

We have shown that quantum field theories obeying the microscopic phase space condition of
\cite{Bos:short_distance_structure} admit a large class of nontrivial quantum
inequalities: to every classically positive expression, i.e., a sum of
absolute squares, we find a combination of composite fields that is
positive up to an error obeying defined estimates and vanishing in the
short distance limit. The composite fields appearing in such QIs are
smeared with test functions derived from OPE coefficients as well as a
choice of test function $g$. In the free field case, these smearing
functions bore a simple relationship to $g$, at least for the normal
product; here, the relationship is less direct, although we have 
succeeded in classifying their structure under simplifying assumptions
within dilation covariant theories. Our inequalities are primarily valid in
the short-distance limit, when the support of the test functions shrinks to a
point. However, we also discussed how to obtain inequalities for smearing
functions with extended (mesoscopic) support, in which the remainder term
can be reduced at the expense of increasing the contributions from other
composite fields.

To conclude we mention a number of open questions and avenues for
further investigation. First, more progress can be made in understanding
the sampling functions arising. For example, in the dilation covariant
setting, one could also allow general finite-dimensional irreducible
representations of the dilation group. Second, it would probably not be hard to
generalise our bounds from smearing along a fixed timelike inertial curve to
smearing along arbitrary smooth timelike curves in Minkowski space. The
structure of inequalities is not expected to change significantly under
this generalization. Third, one would also like to
establish OPE-based quantum inequalities in curved spacetime. Here, the
situation is complicated by the lack of a global Hamiltonian to specify
scales of spaces of states and fields. A replacement for the topologies
thus induced might be found in the detailed microlocal structure of
$n$-point functions, for example, using wave-front sets modulo Sobolev
regularity (see, e.g., \cite{JunSch:adiabatic}). 
An alternative approach would be to use the stress-energy
tensor as the basis for estimates of high-energy behaviour. Hollands has
recently established an OPE on curved spacetime for perturbatively constructed
theories \cite{Hol:ope_curved}; however, the generalization of the nonperturbative methods used here
presently remains a challenging problem. 

Fourth, it would be desirable to obtain results that directly constrain
the energy density of a quantum field theory, returning to the original
motivation for quantum inequalities. One may
heuristically expect from perturbation theory that the energy density in
purely bosonic theories does
arise from such a sum of squares (although a generalization would be
needed to cater for theories with fermionic fields) and would therefore
be amenable to our approach. However, more direct connections to
the energy density are unknown at present; in fact, the very concept of energy
density is not well established in a nonperturbative context in
purely Minkowski space quantum field theory. More generally, no general nonperturbative
version of the Noether theorem has been found to date. 
In the Wightman framework, only very few results about
pointlike Noether currents are available \cite{Orz:charges,Lop:symmetry}, in
particular an existence proof is missing. In the algebraic framework, partial
results have been achieved \cite{BDL:noether_thm}
on the base of the so-called \emph{split property} of the local algebras
\cite{Dop:superselection_1,DopLon:superselection_2}. 
In effect, it is possible to construct ``local'' energy operators $H_{\ocal,\hat\ocal}$,
which are associated with the observable algebra $\afk(\ocal)$ of a
bounded region $\ocal$ and act like the global Hamiltonian on $\afk(\hat\ocal)$
for a slightly smaller region $\hat\ocal \subset\subset \ocal$. 
These operators fulfil $H_{\ocal,\hat\ocal} \geq 0$, which may be
interpreted as a very weak form of energy inequality: Starting from local
integrals of the energy density, it seems always possible to add appropriate
``boundary terms'', associated with $\ocal \cap \hat\ocal'$, such that the
resulting operator $H_{\ocal,\hat\ocal}$ is positive. However, there is no
explicit control on these boundary terms, not even a means of separating them
from a ``main term'', so that this approach does not yet lead to a meaningful
interpretation in terms of quantum energy inequalities.

In curved spacetime, however, the situation is better. Brunetti,
Fredenhagen and Verch have shown the existence of a stress-tensor in
locally covariant quantum field theories obeying the time-slice axiom \cite{BFV:generally_covariant}.
This stress-energy tensor is obtained by functional differentiation with
respect to metric perturbations. This prevents an immediate identification
of the energy density as a sum of absolute squares of basic fields.
Nevertheless, this may serve as a starting point for future study.

\section*{Acknowledgements}
This work was initiated during the programme `Mathematical and
Physical Aspects of Perturbative Approaches to Quantum Field Theory' at
the Erwin Schr\"odinger Institute, Vienna, and the authors thank the
organisers of the programme and the ESI for financial support. CJF also thanks the
Fakult\"{a}t f\"{u}r Mathematik, University of Vienna for 
hospitality and financial support at various stages of the work. 

It is a pleasure to thank Stefan Hollands for valuable discussions in the early
phases of this work. HB also thanks the II. Institut f\"ur Theoretische
Physik, Hamburg, for hospitality and Klaus Fredenhagen for helpful
remarks. The discussion of positivity of formal power series in
section~\ref{sec:power_series} arose from conversations between HB and
Bernd Kuckert, to whose memory this paper is dedicated.

\bibliographystyle{alpha}
\bibliography{qft}

\begin{thebibliography}{BDM09}

\bibitem[AF03]{AdaFou:sobolev}
Robert~A. Adams and John J.~F. Fournier.
\newblock {\em Sobolev Spaces}, volume 140 of {\em Pure and Applied
  Mathematics}.
\newblock Academic Press, 2nd edition, 2003.

\bibitem[BDL86]{BDL:noether_thm}
Detlev Buchholz, Sergio Doplicher, and Roberto Longo.
\newblock On {Noether}'s theorem in quantum field theory.
\newblock {\em Ann. Phys. (N.Y.)}, 170:1, 1986.

\bibitem[BDM08]{BDM:dilations}
Henning Bostelmann, Claudio D'Antoni, and Gerardo Morsella.
\newblock On dilation symmetries arising from scaling limits.
\newblock arXiv:0812.4762, 2008.

\bibitem[BDM09]{BDM:field_renorm}
Henning Bostelmann, Claudio D'Antoni, and Gerardo Morsella.
\newblock Scaling algebras and pointlike fields. {A} nonperturbative approach
  to renormalization.
\newblock {\em Commun. Math. Phys.}, 285:763--798, 2009.

\bibitem[BFV03]{BFV:generally_covariant}
R.~{Brunetti}, K.~{Fredenhagen}, and R.~{Verch}.
\newblock The generally covariant locality principle - {A} new paradigm for
  local quantum field theory.
\newblock {\em Commun. Math. Phys.}, 237:31--68, 2003.

\bibitem[BJ89]{BucJun:equilibrum}
Detlev Buchholz and Peter Junglas.
\newblock On the existence of equilibrum states in local quantum field theory.
\newblock {\em Commun. Math. Phys.}, 121:255--270, 1989.

\bibitem[Bor64]{Bor:spacelike_smooth}
H.-J. Borchers.
\newblock Field operators as {$C\sp{\infty }$} functions in spacelike
  directions.
\newblock {\em Nuovo Cimento (10)}, 33:1600--1613, 1964.

\bibitem[Bos00]{Bos:operatorprodukte}
Henning Bostelmann.
\newblock {\em Lokale Algebren und Operatorprodukte am Punkt}.
\newblock Thesis, Universit\"at G\"ottingen, 2000.
\newblock Available online at http://webdoc.sub.gwdg.de/diss/2000/bostelmann/.

\bibitem[Bos05a]{Bos:product_expansions}
Henning Bostelmann.
\newblock Operator product expansions as a consequence of phase space
  properties.
\newblock {\em J.~Math. Phys.}, 46:082304, 2005.

\bibitem[Bos05b]{Bos:short_distance_structure}
Henning Bostelmann.
\newblock Phase space properties and the short distance structure in quantum
  field theory.
\newblock {\em J.~Math. Phys.}, 46:052301, 2005.

\bibitem[BP90]{BucPor:phase_space}
Detlev Buchholz and Martin Porrmann.
\newblock How small is the phase space in quantum field theory?
\newblock {\em Ann. Inst. H.~Poincar\'e}, 52:237--257, 1990.

\bibitem[BV95]{BucVer:scaling_algebras}
Detlev Buchholz and Rainer Verch.
\newblock Scaling algebras and renormalization group in algebraic quantum field
  theory.
\newblock {\em Rev. Math. Phys.}, 7:1195--1239, 1995.

\bibitem[BW86]{BucWic:causal_independence}
Detlev Buchholz and Eyvind~H. Wichmann.
\newblock Causal independence and the energy-level density of states in local
  quantum field theory.
\newblock {\em Commun. Math. Phys.}, 106:321--344, 1986.

\bibitem[BW98]{BorWal:formal_gns}
M.~Bordemann and S.~Waldmann.
\newblock Formal {GNS} construction and states in deformation quantization.
\newblock {\em Commun. Math. Phys.}, 195:549--583, 1998.

\bibitem[DF99]{DueFre:qed_observables}
M.~D\"utsch and Klaus Fredenhagen.
\newblock {A local (perturbative) construction of observables in gauge
  theories: The example of {QED}}.
\newblock {\em Commun. Math. Phys.}, 203:71--105, 1999.

\bibitem[DF06]{DawFew:curved_Dirac}
S.~P. Dawson and C.~J. Fewster.
\newblock An explicit quantum weak energy inequality for {D}irac fields in
  curved spacetimes.
\newblock {\em Classical Quantum Gravity}, 23(23):6659--6681, 2006.

\bibitem[DL83]{DopLon:superselection_2}
Sergio Doplicher and Roberto Longo.
\newblock Local aspects of superselection rules~{II}.
\newblock {\em Commun. Math. Phys.}, 88:399--409, 1983.

\bibitem[Dop82]{Dop:superselection_1}
Sergio Doplicher.
\newblock Local aspects of superselection rules.
\newblock {\em Commun. Math. Phys.}, 85:73--86, 1982.

\bibitem[EFV05]{EFV:qm_inequalities}
Simon~P. Eveson, Christopher~J. Fewster, and Rainer Verch.
\newblock Quantum inequalities in quantum mechanics.
\newblock {\em Ann. Henri Poincar{\'e}}, 6:1--30, 2005.

\bibitem[EGJ65]{EGJ:nonpos}
H.~Epstein, V.~Glaser, and Arthur Jaffe.
\newblock Nonpositivity of the energy density in quantized field theories.
\newblock {\em Nuovo Cimento}, 36:1016--1022, 1965.

\bibitem[FE98]{FewEve:qei_flat}
Christopher~J. Fewster and Simon~P. Eveson.
\newblock Bounds on negative energy densities in flat spacetime.
\newblock {\em Phys. Rev. D}, 58:084010, 1998.

\bibitem[Few00]{Few:general_wordline}
Christopher~J. Fewster.
\newblock A general worldline quantum inequality.
\newblock {\em Class. Quant. Grav.}, 17(9):1897--1911, 2000.

\bibitem[Few06]{Few:Bros}
Christopher~J. Fewster.
\newblock Quantum energy inequalities and stability conditions in quantum field
  theory.
\newblock In A.~Boutet~de Monvel, D.~Buchholz, D.~Iagolnitzer, and
  U.~Moschella, editors, {\em Rigorous Quantum Field Theory: A Festschrift for
  Jacques Bros}, volume 251 of {\em Progress in Mathematics}, pages 95--111,
  Boston, 2006. {Birkh\"auser}.

\bibitem[Few07]{Few:categorical}
Christopher~J. Fewster.
\newblock {Quantum energy inequalities and local covariance. II: Categorical
  formulation}.
\newblock {\em Gen. Rel. Grav.}, 39:1855--1890, 2007.

\bibitem[FH81]{FreHer:pointlike_fields}
Klaus Fredenhagen and Joachim Hertel.
\newblock Local algebras of observables and pointlike localized fields.
\newblock {\em Commun. Math. Phys.}, 80:555--561, 1981.

\bibitem[FH87]{FreHaa:covariant_scaling}
Klaus Fredenhagen and Rudolf Haag.
\newblock Generally covariant quantum field theory and scaling limits.
\newblock {\em Commun. Math. Phys.}, 108:91--115, 1987.

\bibitem[FH05]{FewHol:qei_conformal}
Christopher~J. Fewster and Stefan Hollands.
\newblock Quantum energy inequalities in two-dimensional conformal field
  theory.
\newblock {\em Rev. Math. Phys.}, 17:577--612, 2005.

\bibitem[FHR02]{FRH:spatial_average}
L.~H. Ford, Adam~D. Helfer, and Thomas~A. Roman.
\newblock {Spatially averaged quantum inequalities do not exist in
  four-dimensional spacetime}.
\newblock {\em Phys. Rev.}, D66:124012, 2002.

\bibitem[Fla97]{Fla:qi_twodimensional}
\'Eanna~\'E. Flanagan.
\newblock Quantum inequalities in two-dimensional minkowski spacetime.
\newblock {\em Phys. Rev. D}, 56(8):4922--4926, Oct 1997.

\bibitem[FM03]{FewMis:flat_Dirac}
C.~J. Fewster and B.~Mistry.
\newblock Quantum weak energy inequalities for the {D}irac field in flat
  spacetime.
\newblock {\em Phys. Rev. D (3)}, 68(10):105010, 6, 2003.

\bibitem[FO08]{FewOst:qei_nonmin}
Christopher~J. Fewster and Lutz~W. Osterbrink.
\newblock {Quantum Energy Inequalities for the Non-Minimally Coupled Scalar
  Field}.
\newblock {\em J. Phys.}, A41:025402, 2008.

\bibitem[FOP05]{FOP:p_nuclearity}
Christopher~J. Fewster, Izumi Ojima, and Martin Porrmann.
\newblock {$p$}-nuclearity in a new perspective.
\newblock {\em Lett. Math. Phys.}, 73(1):1--15, 2005.

\bibitem[For78]{For:quantum_coherence}
L.~H. Ford.
\newblock {Quantum Coherence Effects and the Second Law of Thermodynamics}.
\newblock {\em Proc. Roy. Soc. London A}, 364:227--236, December 1978.

\bibitem[{For}91]{For:neg_energy_flux}
L.~H. {Ford}.
\newblock {Constraints on negative-energy fluxes}.
\newblock {\em Phys. Rev. D}, 43:3972--3978, June 1991.

\bibitem[FP81]{FefPho:gaarding_ineq}
C.~Fefferman and D.~H. Phong.
\newblock The uncertainty principle and sharp {G}{\aa}rding inequalities.
\newblock {\em Comm. Pure Appl. Math.}, 34:285, 1981.

\bibitem[FP03]{FewPfe:qei_spin1}
Christopher~J. Fewster and Michael~J. Pfenning.
\newblock A quantum weak energy inequality for spin-one fields in curved
  space-time.
\newblock {\em J. Math. Phys.}, 44(10):4480--4513, 2003.

\bibitem[FR95]{ForRom:averaged_energy}
L.~H. Ford and Thomas~A. Roman.
\newblock Averaged energy conditions and quantum inequalities.
\newblock {\em Phys. Rev. D}, 51(8):4277--4286, Apr 1995.

\bibitem[FR97]{ForRom:restrictions_flat}
L.~H. Ford and Thomas~A. Roman.
\newblock Restrictions on negative energy density in flat spacetime.
\newblock {\em Phys. Rev. D}, 55(4):2082--2089, Feb 1997.

\bibitem[FR03]{FewRom:null}
Christopher~J. Fewster and Thomas~A. Roman.
\newblock Null energy conditions in quantum field theory.
\newblock {\em Phys. Rev. D (3)}, 67(4):044003, 11, 2003.

\bibitem[FS08]{FewSmi:qei_abs}
Christopher~J. Fewster and Calvin~J. Smith.
\newblock Absolute quantum energy inequalities in curved spacetime.
\newblock {\em Ann. Henri Poincar\'e}, 9(3):425--455, 2008.

\bibitem[FT99]{FewTeo:qei_static}
Christopher~J. Fewster and Edward Teo.
\newblock Bounds on negative energy densities in static space-times.
\newblock {\em Phys. Rev. D (3)}, 59(10):104016, 10, 1999.

\bibitem[FV02]{FewVer:qei_curved}
Christopher~J. Fewster and Rainer Verch.
\newblock A quantum weak energy inequality for dirac fields in curved
  spacetime.
\newblock {\em Commun. Math. Phys.}, 225:331--359, 2002.

\bibitem[FV03]{FewVer:stability}
Christopher~J. Fewster and Rainer Verch.
\newblock Stability of quantum systems at three scales: Passivity, quantum weak
  energy inequalities and the microlocal spectrum condition.
\newblock {\em Commun. Math. Phys.}, 240:329--375, 2003.

\bibitem[GJ87]{GliJaf:quantum_physics}
James Glimm and Arthur Jaffe.
\newblock {\em Quantum Physics -- A functional integral point of view}.
\newblock Springer, New York, 2nd edition, 1987.

\bibitem[GS68]{GelShi:genfun1}
I.~M. Gelfand and G.~E. Shilov.
\newblock {\em Generalized functions}, volume~1.
\newblock Academic Press, New York, 1968.

\bibitem[Haa96]{Haa:LQP}
Rudolf Haag.
\newblock {\em Local Quantum Physics}.
\newblock Springer, Berlin, 2nd edition, 1996.

\bibitem[HLZ06]{HLZ:massless_RaritaSchwinger}
Bo~Hu, Yi~Ling, and Hong-bao Zhang.
\newblock {Quantum inequalities for massless spin-3/2 field in Minkowski
  spacetime}.
\newblock {\em Phys. Rev.}, D73:045015, 2006.

\bibitem[HO96]{HaaOji:germs}
Rudolf Haag and Izumi Ojima.
\newblock On the problem of defining a specific theory within the frame of
  local quantum physics.
\newblock {\em Ann. Inst. H. Poincar\'e}, 64:385--393, 1996.

\bibitem[Hol07]{Hol:ope_curved}
Stefan Hollands.
\newblock The operator product expansion for perturbative quantum field theory
  in curved spacetime.
\newblock {\em Commun. Math. Phys}, 273:1--36, 2007.

\bibitem[HS65]{HaaSwi:compactness}
Rudolf Haag and J.~A. Swieca.
\newblock When does a quantum field theory describe particles?
\newblock {\em Commun. Math. Phys.}, 1:308--320, 1965.

\bibitem[Hud74]{Hud:Wigner}
R.~L. Hudson.
\newblock When is the {W}igner quasi-probability density non-negative?
\newblock {\em Rep. Mathematical Phys.}, 6(2):249--252, 1974.

\bibitem[Joh61]{Joh:green_2d}
K.~Johnson.
\newblock Solution of the equations for the green's functions of a two
  dimensional relativistic field theory.
\newblock {\em Nuovo Cimento}, 20:773--790, 1961.

\bibitem[JS02]{JunSch:adiabatic}
W.~Junker and E.~Schrohe.
\newblock Adiabatic vacuum states on general spacetime manifolds: definition,
  construction, and physical properties.
\newblock {\em Ann. Henri Poincar\'e}, 3(6):1113--1181, 2002.

\bibitem[Lop91]{Lop:symmetry}
Jan Lopuszanski.
\newblock {\em An Introduction to Symmetry and Supersymmetry in Quantum Field
  Theory}.
\newblock World Scientific, Singapore, 1991.

\bibitem[OG03]{OluGra:static_domain_wall}
Ken~D. Olum and Noah Graham.
\newblock Static negative energies near a domain wall.
\newblock {\em Phys. Lett. B}, 554(3-4):175--179, 2003.

\bibitem[Orz70]{Orz:charges}
Claudio~A. Orzalesi.
\newblock Charges and generators of symmetry transformations in quantum field
  theory.
\newblock {\em Rev. Mod. Phys.}, 42(4):381--408, October 1970.

\bibitem[PF98]{PfeFor:qei_static}
Michael~J. Pfenning and L.~H. Ford.
\newblock Scalar field quantum inequalities in static spacetimes.
\newblock {\em Phys. Rev. D (3)}, 57(6):3489--3502, 1998.

\bibitem[Pfe02]{Pfe:qei_em}
Michael~J. Pfenning.
\newblock Quantum inequalities for the electromagnetic field.
\newblock {\em Phys. Rev. D (3)}, 65(2):024009, 13, 2002.

\bibitem[Por04]{Por:pw_disint_ii}
Martin Porrmann.
\newblock Particle weights and their disintegration {II}.
\newblock {\em Commun. Math. Phys.}, 248:305--333, 2004.

\bibitem[RS75]{ReeSim:mmmp2}
Michael Reed and Barry Simon.
\newblock {\em Methods of Modern Mathematical Physics}, volume II: Fourier
  Analysis, Self-Adjointness.
\newblock Academic Press, San Diego, 1975.

\bibitem[Smi07]{Smi:Dirac_abs}
Calvin~J. Smith.
\newblock An absolute quantum energy inequality for the {D}irac field in curved
  spacetime.
\newblock {\em Classical Quantum Gravity}, 24(18):4733--4750, 2007.

\bibitem[SV08]{SchVer:LTE_QEI}
Jan Schlemmer and Rainer Verch.
\newblock {Local Thermal Equilibrium States and Quantum Energy Inequalities}.
\newblock {\em Annales Henri Poincare}, 9:945--978, 2008.

\bibitem[SW64]{StrWig:PCT}
Raymond~F. Streater and Arthur~S. Wightman.
\newblock {\em PCT, Spin and Statistics, and All That}.
\newblock Benjamin, New York, 1964.

\bibitem[Vol00]{Vol:qi_conformal}
Dan~N. Vollick.
\newblock Quantum inequalities in curved two-dimensional spacetimes.
\newblock {\em Phys. Rev. D (3)}, 61(8):084022, 5, 2000.

\bibitem[Wil69]{Wil:non-lagrangian}
Kenneth~G. Wilson.
\newblock Non-lagrangian models of current algebra.
\newblock {\em Phys. Rev.}, 179(5):1499--1512, 1969.

\bibitem[YW04]{YuWu:massive_RaritaSchwinger}
Hong-wei Yu and Pu-xun Wu.
\newblock {Quantum inequalities for the free Rarita-Schwinger fields in flat
  spacetime}.
\newblock {\em Phys. Rev.}, D69:064008, 2004.

\bibitem[Zim70]{Zim:Brandeis}
Wolfhart Zimmermann.
\newblock Local operator products and renormalization in quantum field theory.
\newblock In Stanley Deser, Marc Grisaru, and Hugh Pendleton, editors, {\em
  Lectures on Elementary Particles and Quantum Field Theory}, volume~1,
  Cambridge, 1970. MIT Press.

\end{thebibliography}

\end{document}